\documentclass[journal,twoside,web]{IEEEtran}
\usepackage{generic}
 \usepackage{cite}
 \usepackage{amsmath,amssymb,amsfonts}
\usepackage{algorithmic}
\usepackage{graphicx}
\usepackage{textcomp}
\usepackage{amsmath} 
\usepackage{algorithm}
\usepackage{algorithmic}
\usepackage{ulem}
\newtheorem{assumption}{Assumption}
\newtheorem{theorem}{Theorem}
\newtheorem{lemma}{Lemma}
 \newtheorem{corollary}{Corollary}
 \newtheorem{Definition}{Definition}
 \newtheorem{proposition}{Proposition}
\def\BibTeX{{\rm B\kern-.05em{\sc i\kern-.025em b}\kern-.08em
    T\kern-.1667em\lower.7ex\hbox{E}\kern-.125emX}}
\begin{document}
\title{Learn and Control while Switching: Guaranteed Stability and Sublinear Regret}
\author{Jafar Abbaszadeh Chekan and C\'edric Langbort
\thanks{This paragraph of the first footnote will contain the date on 
which you submitted your paper for review. It will also contain support 
information, including sponsor and financial support acknowledgment. For 
example, ``This work was supported in part by the U.S. Department of 
Commerce under Grant BS123456.'' }
\thanks{J.~A.~Chekan and C.~Langbort (emails: jafar2 \& langbort@illinois.edu) are with the Coordinated
  Science Laboratory and the Department of Aerospace
  Engineering at the University of Illinois at Urbana-Champaign (UIUC).}
}

\maketitle

\begin{abstract}
Over-actuated systems often make it possible to achieve specific performances by switching between different subsets of actuators. However, when the system parameters are unknown, transferring authority to different subsets of actuators is challenging due to stability and performance efficiency concerns. This paper presents an efficient algorithm to tackle the so-called ``learn and control while switching between different actuating modes" problem in the Linear Quadratic (LQ) setting. Our proposed strategy is constructed upon Optimism in the Face of Uncertainty (OFU) based algorithm equipped with a projection toolbox to keep the algorithm efficient, regret-wise. Along the way, we derive an optimum duration for the warm-up phase,  thanks to the existence of a stabilizing neighborhood. The stability of the switched system is also guaranteed by designing a minimum average dwell time. The proposed strategy is proved to have a regret bound of $\mathcal{\bar{O}}\big(\sqrt{T}\big)+\mathcal{O}\big(ns\sqrt{T}\big)$ in horizon $T$ with $(ns)$ number of switches, provably outperforming naively applying the basic OFU algorithm.
\end{abstract}


\section{Introduction}
\label{sec:introduction}
\IEEEPARstart{C}{ontrol} allocation is a vast and richly studied field that addresses the problem of distributing control efforts over redundant actuators to achieve a specific performance. Along this line, supervisory switching control that enables operation with different subsets of actuators is a practical approach for control allocation \cite{ishii2001stabilizing, zhang2020supervisory,staroswiecki2010general, kanellopoulos2019moving}. {In this class of problems, the switching time, the switch mode (which represents a subset of actuators), or both of these parameters can be either specified as part of the general decision-making process or determined by an external unit. Three selected problem classes that illustrate the significance of switching in over-actuated systems are Fault Tolerant Control (FTC), Time-Triggered Protocol (TTP), and moving target defense (MTD). These examples differ in their approach to defining and designing these crucial variables.} {Designing an algorithm to efficiently switch between different actuating modes when system parameters are unknown presents a significant challenge in terms of performance. We refer to this problem as "learn and control while switching." In this context, it is inefficient to disregard previously acquired information about these parameters and start the learning and control process from scratch each time a mode switch occurs. Our goal is to create an algorithm that capitalizes on previously acquired knowledge when commencing a new actuation mode. More specifically, our focus lies on an online Optimism in the Face of Uncertainty (OFU) based algorithm within the framework of Linear Quadratic (LQ) settings. This class of algorithms builds a confidence ellipsoid containing the unknown system parameters, allowing us to design a feedback gain by adopting an optimistic perspective toward this defined set. A naive approach to address our proposed problem would involve the repetitive application of the OFU-based algorithm whenever a mode switch takes place. However, this method is likely to be inefficient, given that it disregards previously acquired knowledge. In this research endeavor, we introduce a novel strategy that leverages learned information. This strategy involves the projection of recently constructed ellipsoids into the mode space that the system has transitioned into. Following this projection, we establish a confidence ellipsoid for control design, incorporating both the knowledge stored within the projected ellipsoid and the online measurements gathered after initiating the actuation. The core technical challenge in constructing this ellipsoid involves confidence set construction while incorporating the already learned information stored within the projected confidence ellipsoid. In our setting, unlike \cite{cohen2019learning}, who use the trace norm of the system parameters for regularization in least square estimation and confidence set construction, we utilize a weighted Frobenius norm. This approach provides us with an opportunity to incorporate the projected ellipsoid into confidence ellipsoid construction. To fully incorporate this idea, we also forgo the confidence ellipsoid normalization applied in \cite{cohen2019learning}. As a result, our confidence ellipsoid representation differs slightly from that of \cite{cohen2019learning}. Regarding the stability guarantee in the switching setting, we need to determine a minimum average dwell time that ensures the system slows down the switching process sufficiently to guarantee boundedness of the state norm during the switching operation.}

Actuator faults can cause a decay in performance (e.g., impact on energy consumption of actuator) or even more devastating incidents \cite{tohidi2016fault}. This kind of safety concerns is often considered within the literature on so-called Fault Tolerant Control (FTC) \cite{blanke2006diagnosis, patton1997fault}, where the goal of switching is to bypass a failing actuator \cite{takahashi2012adaptive}, \cite{yang2015adaptive}. \cite{ouyang2017adaptive}. 

In TTP systems \cite{shaheen2007gateway, blind2013time}, where control must occur through communication channels whose availability varies according to a predetermined/ known a priori schedule, the goal of switching is to conform to the specifications of that protocol. 

As another specific application of switching systems, we can refer to the MTD strategy in cyber-physical security of control systems. MTD, which first appeared in the computer science literature \cite{cho2020toward} is a proactive strategy using which a system can hide its dynamics from a potential attacker by randomly switching between a different subset of actuators. This strategy seems practical considering that adversarial agents have limited resources and cannot compromise all the actuators simultaneously. 
In  \cite{weerakkody2015detecting, weerakkody2016moving}, the authors apply MTD for identifying malicious attacks on sensors and actuators. Furthermore, in \cite{kanellopoulos2019moving} and \cite{griffioen2020moving} the authors show an effort to address this challenge for a system with known dynamics.

One significant difference between these applications is that some are instances of "direct switching" whereas others belong to the "indirect switching" category \cite{narendra1997adaptive}, which impacts the determination of "when and what mode to switch to". For example, in TTP, both the switch time and actuation mode are pre-specified and given. On the other hand, for the FTC, the switch time is simply when an anomaly is detected. For this class of problems, we propose a mechanism that picks the best actuating mode to switch to by examining the richness of learned information up to that point. As for the MTD application, given that high-frequency switches are desirable, we let the system switch as often as desired unless stability is violated. For this aim, we constrain the system to stay in a mode for some minimum average dwell time (MADT). The next actuating method is specified with a similar strategy as in the FTC case. In all cases, the algorithm guarantees closed-loop stability and that parameters are learned well-enough to ensure low regret (defined as the average difference between the closed-loop quadratic cost and the best achievable one, with the benefit of knowing the plant's parameters). The former is addressed by designing a MADT while latter requires detailed insight into available RL algorithms in LQ control literature.

Several studies have recently attempted to address switched systems control under different assumptions when the system model is unknown. Authors in \cite{rotulo2021online} designed data-driven control laws for switched systems with discrete-time linear systems and unknown switch signals. The stability of an unknown switched linear system is investigated probabilistically by observing only finite trajectories in \cite{kenanian2019data}. Furthermore, in \cite{dai2018moments,dai2022convex} design of stabilizing controllers for a switched system is addressed by solely relying on system structure and experimental data and without an explicit plant identification.

Learning-based methods in LQ framework are divided into model-free (\cite{bradtke1994adaptive} \cite {tu2017least, fazel2018global, abbasi2018model, arora2018towards}) and model-based RL approaches. The former ignores parameter estimates and outputs sub-optimal policy by solely relying on the history of data and directly minimizing the cost function  \cite{abbasi2018model}. The latter category usually comes with guaranteed regret, thanks to system identification as a core step in control design. We use similar properties to obtain a guaranteed regret.

A naive approach to model-based adaptive control is a simple philosophy that estimates a system's unknown parameters and treats them as actual parameters in optimal control design. This algorithm is known as certainty equivalence in the literature, and since it decouples the estimation problem from the control one, it may lead to strictly sub-optimal performance \cite{campi1998adaptive}. Model-Based algorithms, in general, can be divided into three main categories, namely "Offline estimation and control synthesis," "Thompson Sampling-based," and "Optimism in the Face of Uncertainty (OFU)" algorithms. OFU that couples estimation and control design procedures have shown efficiency in outperforming the naive certainty equivalence algorithm. Campi and Kumar in \cite{campi1998adaptive} proposed a cost-based parameter estimator, which is an OFU based approach to address the optimal control problem for linear quadratic Gaussian systems with guaranteed asymptotic optimality. However, this algorithm only guarantees the convergence of average cost to that of the optimal control in limit and does not provide any bound on the measure of performance loss for the finite horizon. Abbasi-Yadkori and Szepesvari \cite{abbasi2011online} proposed a learning-based algorithm to address
the adaptive control design of the LQ control system in the finite horizon, with worst case regret bound of $O(\sqrt{T})$. Using $l_2$-regularized least square estimation, they construct high-probability confidence set around unknown parameters of the system and design an algorithm that optimistically plays concerning this set \cite{abbasi2011regret}. Along this line, Ibrahimi and et. al. \cite{ibrahimi2012efficient} later on proposed an algorithm that achieves $O(p\sqrt{T})$ regret bound with state-space dimension of $p$. Authors in \cite{faradonbeh2017optimism} proposed an OFU-based learning algorithm with mild assumptions and $O(\sqrt{T})$ regret. Furthermore, in \cite{chekan2021joint} propose an OFU-based algorithm which for joint stabilization and regret minimization that leverages actuator redundancy to alleviate state explosion during initial time steps when there is low number of data.

The objective function to be minimized in the algorithm \cite{abbasi2011regret} is non-convex which brings about a computational hurdle. However, recently Cohen et al., in \cite{cohen2019learning},  through formulating the LQ control problem in a relaxed semi-definite program (SDP) that accounts for the uncertainty in estimates, proposed a computationally efficient OFU-based algorithm. Moreover, unlike the state-norm bound given by \cite{abbasi2011regret} which is loose, \cite{cohen2019learning} guarantees strongly sequentially stabilizing policies enabling us to come ups with tight upper-bound, which ensures small MADT (appropriate for MTD application). These advantages motivate us to use the relaxed SDP framework in our analysis rather than \cite{abbasi2011regret}.

The remainder of the paper is organized as follows: Section \ref{sec:ProbFor} provides the problem statement and then is followed by Section \ref{sec:prelim} which gives a brief background review, overview of projection technique, and augmentation technique. Moving on to Section \ref{sec:perf}, we first provide an overview of the projection-based strategy. Then, we detail the main algorithmic steps and summarize the stability analysis, which includes MADT design and algorithm performance, particularly in terms of the regret bound. Section \ref{Sec:Extension} briefly discusses the extension of the proposed algorithm to the FTC and MTD type of applications. Afterwards, numerical experiment is given in Section \ref{sec:numEx}. Finally, we conclude the key achievements in Section \ref{sec:Conclusion}. The most rigorous analysis of the algorithms (e.g., stability and regret bounds) and technical proofs are provided in Appendix \ref{partD} while leaving complimentary proofs to Appendix \ref{partF}.
. 
\section{Problem statement}
\label{sec:ProbFor}
Throughout this paper, we use the following notations: $\|A\|$ denotes the operator norm i.e., $\|A\| = \max_{x, \|x\|\leq 1} \|Ax\|$. We denote trace norm  by {$\|\bullet\|_*(=\sqrt{\operatorname{Tr}(\bullet^\top \bullet)})$}. The entry-wise dot product between matrices is denoted by $A \bullet B=\operatorname{Tr}(A^\top B)$. The weighted norm of a matrix $A$ with respect to $V$ is interchangeably denoted by $\operatorname{Tr}(A^\top V A)$ or $\|A\|_{V}^2$. We denote a $m$ dimensional identity matrix by $I_m$. {We denote the ceiling function by $\lceil x \rceil $ that maps $x$ to the least integer greater than or equal to $x$.}

Consider the following time-invariant discrete time LQR system 
\begin{align}
x_{t+1} &=A_{*}x_{t} + B_{*}u_{t}+\omega_{t+1}\label{eq:dyn_atttt} \\
c_{t} &=x_{t}^\top Qx_{t} + u_{t}^\top Ru_{t}\label{eq:obs}
\end{align}
where $A_*\in\mathbb{R}^{n\times n}$, $B_*\in\mathbb{R}^{n\times m}$ are unknown and $Q\in  \mathbb{R}^{n \times n} $ and $R\in  \mathbb{R}^{m \times m} $ are known and respectively positive definite matrices. 

Let $\mathbb{B}$ be the set of all columns, $b^i_*$ ($i\in \{1,...,m\}$) of $B_*$. An element of its power set $2^{\mathbb{B}}$ is a subset $\mathcal{B}_{j}$, $j\in \{1,..., 2^m\}$ of columns corresponding to a submatrix $B_*^j$ of $B_*$ and mode $j$. For simplicity, we assume that $B^1_*=B_*$ i.e., that the first mode contains all actuators. {With this definition in place we can express the system dynamics as follows:}

\begin{align}
	x _{t+1} =A_*x_t+B_*^{\sigma(t)}u_t^{\sigma(t)}+\omega_{t+1} \label{eq:dynam_by_theta1} 
\end{align}
or equivalently
\begin{align}
	x _{t+1} ={\Theta^{\sigma(t)}_{*}}^\top z_{t}+\omega_{t+1}, \quad z_t=\begin{pmatrix} x_t \\ u^{\sigma(t)}_t \end{pmatrix}. \label{eq:dynam_by_theta} 
\end{align}

where $\Theta^{\sigma(t)}_*=(A_*,B^{\sigma(t)}_*)^\top$ is controllable and {$\sigma:{0}\;\cup \mathbb{N}\rightarrow \mathcal{M}$ is a given right-continuous piecewise constant switching signal where $\mathcal{M}$ denotes the set of all controllable modes $(A_,B^{i}_*)$}.

The associated cost with this mode is
\begin{align}
		c^{\sigma(t)}_{t} &=x_{t}^\top Q x_{t} + {u^{\sigma(t)}_t}^\top R^{\sigma(t)} u^{\sigma(t)}_t \label{eq:obsswitch}
\end{align}
where $R^{\sigma(t)}\in \mathbb{R}^{d_{\sigma(t)}\times d_{\sigma(t)}}$ is a block of $R$ which penalizes the control inputs of the actuators of mode $\sigma(t) \in\mathcal{M}.$

{Since we are considering arbitrary, externally-determined, switching patterns (which could include switching to a particular mode
at some time instant, then never switching out of it), controllability of the system in each mode is necessary to ensure controllability of the switched system. This mode-by-mode controllability assumption is also useful to enable OFU-style learning at each epoch.}

The noise process satisfies the following assumption which is a standard assumption in controls community (e.g., \cite{abbasi2011regret}, \cite{cohen2019learning} and \cite{lale2020explore}).  

\begin{assumption}

\label{Assumption 1}
There exists a filtration $\mathcal{F}_{t}$ such that

$(1.2)$ $\omega_{t+1}$ is a martingale difference, i.e., $\mathbb{E}[\omega_{t+1}|\;\mathcal{F}_{t}]=0$

$(1.2)$ $\mathbb{E}[\omega_{t+1}\omega_{t+1}^\top|\;\mathcal{F}_{t}]=\bar{\sigma}_{\omega}^2I_{n}$ for some $\bar{\sigma}_{\omega}^2>0$;

$(1.3)$ $\omega_{t}$ are component-wise sub-Gaussian, i.e., there exists $\sigma_{\omega}>0$ such that for any $\gamma \in \mathbb{R}$ and $j=1,2,...,n$
\begin{align*}
\mathbb{E}[e^{\gamma\omega_{j}(t+1)}|\;\mathcal{F}_{t}]\leq e^{\gamma^2\sigma_{\omega}^2/2}.
\end{align*}
\end{assumption}

Our goal is to design a strategy which can efficiently transfer the control authority between different subsets of actuators when the system parameters are unknown. 

Relying on the past measurements, the problem is designing a sequence $\{u^{\sigma(t)}_t\}$ of control input such that the average expected cost
\begin{align}
	J(u_1, u_2,...) =\lim_{T\to\infty} \frac{1}{T}  \sum _{t=1}^{T} \mathbb {E} [c^{\sigma(t)}_{t}]. \label{eq:Ave_cost} 
\end{align}

{is minimized where $\sigma(t)\in \mathcal{I}$, where $\mathcal{I}=[i_{\tau_1}\; i_{\tau_2}\;...i_{\tau_{ns}}]$, represents the switch sequence. Each element in the set $\mathcal{I}$ specifies both the mode index $i_{\tau_j}\in \mathcal{M}$ and the time $\tau_j$ at which the switch to that mode occurs.} {In the context of TTP applications, this sequence is typically predetermined, and thus the system has access to both the switch times and the corresponding modes. In contrast, for FTC, while the switch times $\tau_j$ are determined by an anomaly detector, the decision regarding which mode to switch to is made by the algorithm itself. Therefore, the algorithm plays a pivotal role in selecting the next mode, denoted as $i_{\tau_j}\in\mathcal{M}$. When the system parameters are known, it may be relatively straightforward to specify the next mode that results in the lowest average expected cost. However, this task becomes significantly more challenging in the absence of such information. The MTD application introduces an even greater level of complexity, as the algorithm must autonomously decide both when and to which mode to switch. In this work, our primary focus is on TTP setting. In Section \ref{Sec:Extension} we will provide a brief overview of the general solutions for the other two settings.}

As a measure of optimality loss having to do with the lack of insight into the model, the regret of a policy is defined as follows: 
\begin{align}
	R_T =\sum _{t=1}^{T} \big(\mathbb {E} [c_{t}^{\sigma(t)}]-J_*(\Theta_*^{\sigma(t)}, Q, R^{\sigma(t)})\big)\label{eq:Reg} 
\end{align}
where $J_*(\Theta_*^{\sigma(t)}, Q, R)$ is the optimal average expected cost of actuating in mode $\sigma(t)\in \{1,..., 2^m\}$.

{\begin{proposition}\label{prop:linRegr}
The regret incurred by employing a fixed feedback controller $K_0$ on a LQ system over $T_0$ rounds follows an order of $\mathcal{O}(T_0)$.
\end{proposition}}

\begin{assumption}
$ $
\begin{enumerate}
\item There are known constants $\alpha_0, \alpha_1, \sigma, \vartheta, \nu>0$ such that,
\begin{align*}
& \alpha_0I\leq Q\leq \alpha_1 I,\; \; \; \alpha_0I\leq R^i\leq \alpha^i_1 I
\\
& W=\mathbb{E}[\omega_{t}\omega_{t}^\top]=\sigma^2I,\; \; \;\|\Theta _{*}^i\|\leq \vartheta_i,\; \; \; J_*(\Theta_*^{i}, Q, R)\leq \nu_i 
\end{align*}
for all $i \in \mathcal{M}$.
\item There is an initial stabilizing policy $K_0$ to start algorithm with.
\end{enumerate}
\end{assumption}

Note that having access to $\vartheta_i$ and $\nu_i$ for all $i$ separately is quite demanding and, hence, in practice, we only require knowledge of $\bar{\vartheta}:=\max_{i}\vartheta_i$ and $\bar{\nu}=\max_{i}\nu_i$. {The assumption of having access to an initial stabilizing policy like $K_0$ is a common practice in the literature (as seen in \cite{cohen2019learning} and \cite{dean2018regret}). Nevertheless, it is possible to develop a stabilizing policy in real-time by relying on some a priori bounds. For further details on this approach, we encourage readers to explore \cite{lale2020explore} and \cite{chekan2021joint}.}

Developing a "learn and control while switching" strategy requires overcoming two core challenges: performance and stability. We propose an efficient OFU-based algorithm equipped with a projection tool that can efficiently transfer the learned information from an actuating mode to others to achieve the former goal. This strategy outperforms the naive approach of repeatedly applying the standard OFU algorithm. {Fast recurrent switching between the actuating modes can cause explosion of state norm over finite horizon.} {Hence, in TTP setting, we operate under the assumption that the switches occur at a sufficiently slow rate on average, with the time interval between two consecutive switches $\tau_{j+1}-\tau_j$ being longer than some minimum average dwell time $\tau_{MAD}$.} 

Before presenting the strategy in more detail (and, in a later section, the specifics of the algorithms), we start by reviewing the main ingredients of our approach to build the foundation for our analysis. 

\section{Preliminaries}
\label{sec:prelim}


\subsection {Primal SDP formulation of LQR problem}

Consider a linear system
\begin{align}
x_{t+1} &=Ax_{t} + Bu_{t}+\omega_{t+1}\label{eq:dyn_atttt2} \\
c_{t} &=x_{t}^\top Qx_{t} + u_{t}^\top Ru_{t}.\label{eq:obs2}
\end{align}

Observing that a steady-state joint state and input distribution $(x_{\infty},u_{\infty})$ exists for any stabilizing policy $\pi$, we denote the covariance matrix of this joint distribution as $\mathcal{E}(\pi)$, which is defined as follows:
\begin{align*}
\mathcal{E}(\pi)= \mathbb{E}\begin{pmatrix}
x_{\infty}x_{\infty}^\top & x_{\infty}u_{\infty}^\top \\
u_{\infty}x_{\infty}^\top & u_{\infty}u_{\infty}^\top
\end{pmatrix}
\end{align*}

where $u_{\infty}=\pi(x_{\infty})$ and the average expected cost of the policy is given as follows

\begin{align*}
J(\pi)= \begin{pmatrix}
Q & 0 \\
0 & R
\end{pmatrix}\bullet \mathcal{E}(\pi).
\end{align*}

Now, for a linear stabilizing policy $K$ that maps $u_{\infty}=Kx_{\infty}$, the covariance matrix of the joint distribution takes the following form:

\begin{align}
&\mathcal{E}(K)=\begin{pmatrix}
X & XK^\top \\
KX & KXK^\top
\end{pmatrix}. \label{eq:good123er}
\end{align}

where $X=\mathbb{E}[x_{\infty}x_{\infty}^\top]$. Then the average expected cost is computed as follows 

\begin{align*}
J(K)= \begin{pmatrix}
Q & 0 \\
0 & R
\end{pmatrix}\bullet \mathcal{E}(K)=(Q+K^\top R K)\bullet X.
\end{align*}

Now given $\Theta=(A\; B)^\top$, the LQR control problem can be formulated in a SDP form, as follows:

\begin{align}
\nonumber\textrm{minimize}_{\Sigma}\; \; \; \begin{pmatrix}
Q & 0 \\
0 & R
\end{pmatrix}\bullet \Sigma\\
\nonumber\textrm{S.t.}\; \; \; \Sigma_{xx}={\Theta}^\top\Sigma\Theta+W\\
\Sigma>0\label{eq:SDPKhali}
\end{align}

where any feasible solution $\Sigma$ can be represented as follows:

\begin{align}
  \Sigma=\begin{pmatrix}
\Sigma_{xx} & \Sigma_{xu} \\
\Sigma_{ux} & \Sigma_{uu}
\end{pmatrix}  
\end{align}

in which $\Sigma_{xx} \in  \mathbb{R}^{n \times n}$, $\Sigma_{uu} \in  \mathbb{R}^{m \times m}$, and $\Sigma_{ux}=\Sigma_{xu} \in  \mathbb{R}^{m \times n}$. The optimal value of (\ref{eq:SDPKhali}) coincides with the average expected cost $J_*(\Theta, Q,R)$, and for $W>0$ (ensuring $\Sigma_{xx}>0$), the optimal policy of the system, denoted as $K_*(\Theta, Q,R)$ and associated with the optimum solution $\Sigma^*(\Theta, Q,R)$, can be derived from $K_*(\Theta, Q,R)=\Sigma^*_{{ux}}(\Theta, Q,R)\Sigma^{*^{-1}}_{{xx}}(\Theta, Q,R)$ considering (\ref{eq:good123er}).

In \cite{cohen2019learning}, the relaxed form of the SDP is used to design the stabilizing control policy when the matrix $\Theta$ is replaced by its estimates (i.e., confidence set of $\Theta$). We will slightly modify the relaxed SDP formulation for our purpose. 

\subsection {System Identification}
\label{sysIdentnai}

The proposed algorithm includes a system identification step that applies self-normalized process to obtain least square estimates and uses the martingale difference property of process noise to construct high probability confidence set around the system's parameters. To sketch a brief outline, assume the system actuates in the first mode and evolves in an arbitrary time interval $[0, \;t]$. One can then simply rewrite (\ref{eq:dyn_atttt}) as follows:  
\begin{align}
X_t &=Z_t \Theta_{*}+W_t \label{compactdyn}
\end{align}
where $X_t$ and $Z_t$ are matrices whose rows are $x^\top_{1}, ...,x^\top_t$ and ${z}^\top_{0}, ...,{z}^\top_{t-1}$ respectively. And in a similar way $W_t$ is constructed by concatenating  $\omega_{1}^\top,...,\omega_t^\top$ vertically. 

Using the self-normalized process, the least square estimation error, $e(\Theta)$ can be defined as
 \begin{align}
 \nonumber e(\Theta)&=\lambda \operatorname{Tr} \big((\Theta-\Theta_0)^\top(\Theta-\Theta_0)\big)+\\
 &\quad\sum _{s=0}^{t-1} \operatorname{Tr} \big((x_{s+1}-\Theta^\top z_{s})(x_{s+1}-\Theta^\top z_{s})^\top)\big) \label{eq:LSE_pr}
 \end{align}
 where $\Theta_0$ is some given initial estimate and {$\lambda>0$} is a regularization parameter.
 This yields the $l^{2}$-regularized least square estimate:
 \begin{align}
 \hat{\Theta}_{t} &=\operatorname*{argmin}_\Theta e(\Theta)=V_{[0,\;t]}^{-1}\big(Z_t^\top X_t+ \lambda \Theta_0\big). \label{eq:LSE_Sol12} 
 \end{align}
where 
 \begin{align*}
V_{[0,\;t]}=\lambda I + \sum_{s=0}^{t-1} z_{s}z_{s}^T=\lambda I +Z_t^\top Z_t,
 \end{align*}
is covariance matrix. Then assuming $\|\Theta_*-\Theta_0\|_*\leq \epsilon$ for $\epsilon >0$, a high probability confidence {set containing unknown parameters of system $\Theta_*$ is constructed as}
 \begin{align}
\nonumber \mathcal{C}_t(\delta):=\big\{&\Theta \in \mathbb{R}^{(n+m)\times n}|\\
&\operatorname{Tr}\big((\Theta-\hat{\Theta}_t)^TV_{[0,\;t]}(\Theta-\hat{\Theta}_t)\big) &\leq r_t\big\} \label{eq:confSet1_tighterghfff}
\end{align}
where 
\begin{align}
    r_t=\bigg( \sigma_{\omega} \sqrt{2n \log\frac{n\det(V_{[0,\;t]}) }{\delta \det(\lambda I)}}+\sqrt{\lambda} \epsilon \bigg)^2. \label{radius_centralEl_realTime}
\end{align}

It can be shown that the true matrix $\Theta_*$ is guaranteed to belong to this set with probability at least $1-\delta$, for $0<\delta<1$. The right hand side of (\ref{eq:confSet1_tighterghfff}) is dependent on the user-defined parameters $\lambda$, and initial estimation error $\epsilon$. The proof follows a combination of steps of \cite{cohen2019learning} and \cite{abbasi2011regret} with the difference that we ignore the normalization of the confidence set. For the sake of brevity, we may sometimes use $\Delta_t=\Theta-\hat{\Theta}_t$ in our analysis.

\subsection{$(\kappa,\gamma)-$ strong sequential stability}
We review the notation of strong sequentially stabilizing introduced in \cite{cohen2019learning}, which will be used in the context of switched system as well. 

{\begin{Definition} \label{def:stability}
Consider a linear plant parameterized by $A$ and $B$. The closed-loop system matrix $A+BK_t$ is $(\kappa, \gamma)-$ strongly stable for $\kappa>0$ and $0<\gamma<1$ if there exists $H_t\succ 0$ and $L_t$ such that $A_*+B_*K_t=H_tL_t{H}^{-1}_t$ and
\begin{enumerate}
    \item $\|L_t\|\leq 1-\gamma$ and $\|K_t\|\leq \kappa$
     \item $\|H_t\|\leq B_0$ and $\|{H}_t^{-1}\|\leq 1/b_0$ with $\kappa=B_0/b_0$
\end{enumerate}
A sequence $A+BK_1, A+BK_2, ...$ of closed-loop system matrices is $(\kappa, \gamma)-$strongly sequentially stable if each $A+BK_t$ is $(\kappa, \gamma)-$ strongly stable and  $\|H^{-1}_{t+1}H_t\|\leq 1+\gamma/2$.
$     $
$     $
Furthermore, we say a sequence $K_1, K_2,...$ of control gains is $(\kappa, \gamma)-$strongly sequentially stabilizing for a plant $(A, B)$ if the sequence $A+BK_1, A+BK_2, ...$ is $(\kappa, \gamma)-$strongly sequentially stable.
\end{Definition}}

    
     


As shown in \cite{cohen2019learning} applying such a strongly stabilizing  sequence of controllers starting from state $x_{t_0}$ guarantees the boundedness of state trajectory as follows:
 \begin{align}
\|x_t\|\leq \kappa_i e^{-\gamma_i (t-1)/2}\|x_{t_0}\|+\frac{2\kappa_i}{\gamma_i} \max\limits_{t_0\leq s\leq t} \|w_s\|.\label{eq:rsp_seq_st} 
 \end{align}

\begin{Definition} \label{Dref:neigborhood}
The set $\mathcal{N}_{(\kappa_0, \gamma_0)}(\Theta_*)=\{\Theta\in \mathbb{R}^{(n+m)\times n}:\; \|\Theta-\Theta_*\|\leq \epsilon\}$, for some $\epsilon>0$, is $(\kappa_0,\gamma_0)$-stabilizing neighborhood of the system (\ref{eq:dyn_atttt}-\ref{eq:obs}). In this neighborhood, for any $\Theta^{\prime}\in \mathcal{N}_{(\kappa_0, \gamma_0)}(\Theta_*)$, the designed $(\kappa^{\prime},\gamma^{\prime})$-strongly stabilizing controller $K_*(\Theta^{\prime}, Q, R)$  applied to the system (\ref{eq:dyn_atttt}), achieves an average expected cost lower than that of $K_0$.
\end{Definition}

 \subsection {Projection of the Confidence Ellipsoid}\label{Sec_Proj}

Suppose at time $t$ the central confidence ellipsoid 

\begin{align}
\nonumber \mathcal{C}_t(\delta):=\{&\Theta\in \mathbb{R}^{(n+m)\times n}|\\ &\operatorname{Tr}\big((\Theta-\hat{\Theta}_{t})^\top V_{[0,\;t]} (\Theta-\hat{\Theta}_{t})\big) &\leq r_t\} \label{eq:Mode1Ellips}
\end{align}
is given. By projecting the central ellipsoid  (\ref{eq:Mode1Ellips}), rich with information until time $t$, into the space of mode $i$, the corresponding confidence set of parameters for mode $i$ is again an ellipsoid, denoted by $\mathcal{C}^{pi}_t(\delta)$ and represented as follows:
\begin{align}
\nonumber \mathcal{C}^{pi}_t(\delta):=\{&\Theta^i\in \mathbb{R}^{(n+m_i)\times n}|\\
&\operatorname{Tr}\big((\Theta^i-\hat{\Theta}_{t}^{pi})^\top V^{pi}_{[0\; t]} (\Theta^i-\hat{\Theta}_{t}^{pi})\big) &\leq r\}. \label{eq:projected}
\end{align}

The characteristics $\hat{\Theta}_{t}^{pi}$ and $V^{pi}_{[0\; t]}$ of this projected ellipsoid can be computed following the steps specified below.

Consider the  ellipsoid
\begin{align}
\operatorname{Tr}(\Delta^TV\Delta)\leq r \label{eq:ellipsToProj}
\end{align}
where $\Delta=\Theta-\hat{\Theta}$ and $\Theta$ has $n+m$ rows. We aim to project this set to a lower dimensional space of mode $i$ identified by $\Theta^i \in \mathbb{R}^{(n+m_i)\times n}$ for some $m_i<m$. To this end, we first change the order of rows in such a way that the rows corresponding to the actuators of mode $i$, $\tilde{\Delta}_1\in \mathbb{R}^{(n+m_i)\times n}$ come first whereas $\tilde{\Delta}_2\in \mathbb{R}^{(m-m_i)\times n}$. The reformulated ellipsoid then is written as follows:
\begin{align}
\operatorname{Tr}(\tilde{\Delta}^\top \tilde{V}\tilde{\Delta})\leq r,\quad \tilde{\Delta}=\left[
\begin{array}{c}
\tilde{\Delta}_1 \\
\hline
\tilde{\Delta}_2 
\end{array}
\right]. \label{eq:ellipsToProjre}
\end{align}

{Where $\tilde{V}$ is in $\mathbb{R}^{(n+m)\times (n+m)}$ similar to $V$.} Let us now rewrite the parametrization of ellipsoid (\ref{eq:ellipsToProjre}) as follows

\begin{align*}
\operatorname{vect}(\tilde{\Delta})^T\tilde{G} \operatorname{vect} (\tilde{\Delta})\leq r 
\end{align*} 
{where $\operatorname{vect}(\tilde{\Delta})$ in  $\mathbb{R}^{(n+m) n}$ is a vectorization of $\tilde{\Delta}$ and $\tilde{G}=I_{(n\times n)}\otimes \tilde{V}$.}

Suppose $\partial E_{pi}$ is the boundary of the projection of the ellipsoid  $E:=\operatorname{vect}(\tilde{\Delta})^T\tilde{G} \operatorname{vect}(\tilde{\Delta})\leq r$ on our target lower dimensional space. It is obvious that any point $Y\in \partial E_{pi}$ is the projection of a point $X\in \partial E$ and both of the points are on the tangent space of $\partial E$ parallel to the direction of ${\tilde{\Delta}_2}$. Hence, we can deduce that on these points, $\nabla_{\tilde{\Delta}_2}E=0$ which gives $m-m_i$ algebric equations. By solving these equation for the components of ${\tilde{\Delta}_2}$s in terms of the components of ${\tilde{\Delta}_1}$, and subsequently substituting them into $E$ we obtain the projected ellipsoid $E_{pi}$ which can be written in following compact form :
\begin{align}
\operatorname{vect}(\tilde{\Delta}_1)^T(U-M^TT^{-1}M) \operatorname{vect}(\tilde{\Delta}_1)\leq r. \label{eq:projected2}
\end{align}
where $U$, $M$, and $T$ are the block matrices of the partitioned $\tilde{G}$, e.g.,
\begin{align}
\tilde{G}=
\left[
\begin{array}{c|c}
U & M^T \\
\hline
M & T
\end{array}
\right]. \label{eq:partitioned}
\end{align} 
We can easily construct this partition noting that {$U\in \mathbb{R}^{n(n+m_i)\times n(n+m_i)}$}. By sorting the arrays of vector $\operatorname{vect}(\tilde{\Delta}_1)$ by order into a matrix $\Delta^{pi}$ of dimension $\mathbb{R}^{(n+m_i)\times n}$ the projected confidence ellipsoid (\ref{eq:projected2}) can be rewritten in the standard format of
\begin{align*}
\operatorname{Tr}({\Delta^{pi}}^TV^{pi}\Delta^{pi})\leq r.
\end{align*}
where covariance matrix {$I_{(n\times n)}\otimes V^{pi}=U-M^\top T^{-1}M$}.

\subsection{Augmentation}
\label{augument}
Let us call the mode consisting of all actuators (i.e., mode $i=1$) the "central mode". Its associated confidence ellipsoid $\mathcal{C}_t(\delta)$ contains all information up to time $t$ and can be computed at all times even when the system operates in a different mode. Indeed when the system is in the mode $i$ with the following dynamics:
\begin{align}
	x _{t+1} ={\Theta _{*}^i}^\top z^i_t+\omega_{t+1}, \quad z^i_t=\begin{pmatrix} x_t \\ u^i_t \end{pmatrix} \label{eq:dynam_by_theta2} 
\end{align}
we can  equivalently rewrite it as follows:

\begin{align}
	x _{t+1} ={\Theta _{*}}^\top z_t+\omega_{t+1}, \quad z_t=\begin{pmatrix} x_t \\ u_t \end{pmatrix} \label{eq:dynam_by_theta3} 
\end{align}
where $u_t \in \mathbb{R}^{m}$ and 
\begin{align}
  u_t(\mathcal{B}_i)= u_t^i,\;
  u_t(\mathcal{B}_i^c)=0
\end{align}
with $\mathcal{B}_i^c$ being the complement of $\mathcal{B}_i$.

In turn the estimation process of \ref{sysIdentnai} can be performed on (\ref{eq:dynam_by_theta3}) to obtain updated estimates of $\Theta_*$ computed by the central ellipsoid $\mathcal{C}_t(\delta)$. 

We refer to this simple mathematical operation as "augmentation technique" using which we can update central ellipsoid regardless of the actuation mode.

\section{Projection-based strategy}
\subsection{Overview}\label{CC}
A possible way to tackle the "learn and control while switching problem:" is to apply the OSLO algorithm described in \cite{cohen2019learning} repeatedly, each time a switch occurs. However, such a naive approach is likely inefficient because (1) it ignores previously learned information and, (2) it runs a warm-up phase after each switch. As a (provably) better alternative, we propose a projection-based strategy which relies on the central ellipsoid to incorporate previously learned information in control design procedure. A controller that plays optimistically with respect to this set is then computed by solving a relaxed SDP. 

To clarify the main idea, consider a simple example of control system with three actuators which offers seven actuating modes assuming that they all satisfy the controllability condition. Let $M_1$ denote the mode 1 which takes advantage of all the  actuators. The modes $\{M_2, M_3, M_4\}$ and {$\{M_5, M_6, M_7\}$} are the ones with two actuators and single actuator respectively. We can consider $M_1$ as the central mode which is at the same time the parent of all other actuating modes. {Figure 1 illustrates the relationships between these actuating modes.} 

\begin{figure}[thpb]
	\centering
	\vspace{5pt}
	\includegraphics[scale=.4]{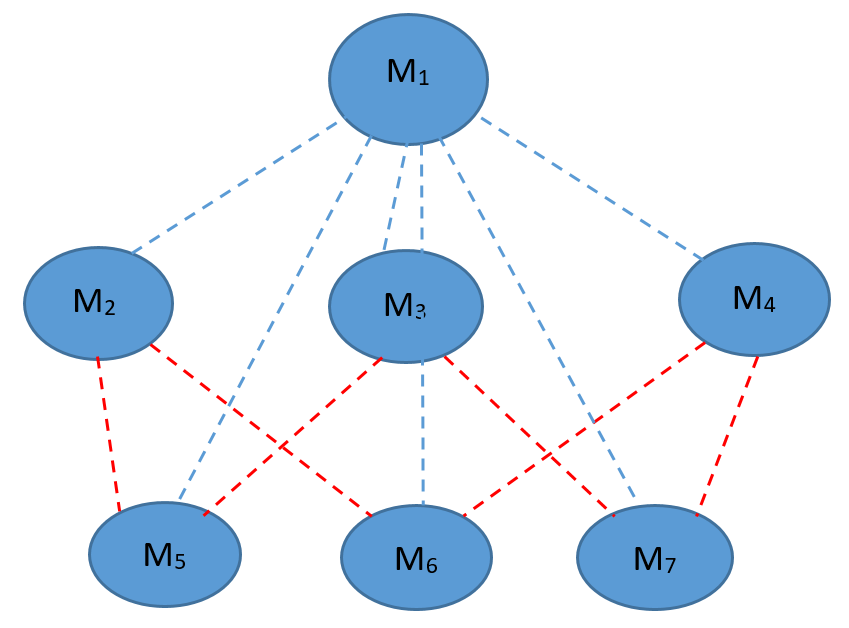}
	\caption{Parent-Child Dependency}
	\label{shem_fig}
\end{figure}
{The dashed lines indicate how the modes are interconnected. For instance, the column set of modes $M_5$ and $M_6$ are the subsets of the column set of mode $M_2$.} We break the strategy into two phases, the warm-up and Learn and Control while Switching (LCS) phases, whose pseudo codes have been provided in Algorithms 1 and 2. The strategy starts with Algorithm 1, which applies a given $(\kappa_0,\gamma_0)-$strong stabilizing policy feedback $K_0$ for $T_0$ time steps to provide Algorithm 2 with close enough parameter estimates $\Theta_0$. For warm-up duration $T_0$ computation, we use the recent results of \cite{mania2019certainty} on the existence of a stabilizing neighborhood around unknown parameters of the system. We compute $T_0$ such that the controller designed by Algorithm 2 at time $T_0$, which relies on the estimate $\Theta_0$, performs better than the initial stabilizing controller $K_0$ and ensures the achievement of a favorable overall regret. The computation of this duration is detailed in Theorem \ref{thm:WarmUp_Duration}. 


After the warm-up phase, Algorithm 2 is applied whenever a switch occurs. This algorithm possesses the capability to access the updated central ellipsoid, which it maintains through the utilization of the augmentation technique. LCSA projects the central ellipsoid to the space of the actuating mode $M_i$; the mode system operates in for the epoch (time interval between two subsequent switches of actuating mode). During operation in an epoch, the algorithm applies OFU for the actuating mode $M_i$ to obtain the control input. In parallel, it updates (i.e., enriches) the central ellipsoid by using the augmentation technique. 

In the next section, we propose LCSA algorithm for the TTP application where the switch sequence (both switch time and modes) are given upfront. We define this sequence by $\mathcal{I}=[i_{\tau_1}\; i_{\tau_2}\;...i_{\tau_{ns}}]$ where $i_{\tau_i} \in \mathcal{M}$ and $\tau_i$ are switch times and $ns$ is total number of switches. Having $\Theta_0$ provided by Algorithm 1, we reset the time and let switch time $\tau_i\in \{\tau_1, \tau_2,..., \tau_{ns}\}$ where $\tau_1=0$. We denote the time interval between two subsequent switches, i.e., $[\tau_k,\tau_{k+1}]$ as epoch $k$. This algorithm, however, can be slightly modified to address FTC and MTD classes of applications. We will discuss these modifications in Section \ref{Sec:Extension}.

\begin{algorithm} 
	\caption{\small Warm-Up Phase \normalsize} \label{alg:IExp}
		\begin{algorithmic}[1]
		\STATE \textbf{Inputs:}$T_{0}$$\,\vartheta>0,$$\,\delta>0,$$\, \sigma_{\omega},\, \eta_t \sim \mathcal{N}(0, 2\sigma_{\omega}^2\kappa_0^2 I)$
		
		\STATE set $\lambda=\sigma_{\omega}^2\vartheta ^{-2}$, $V_0=\lambda I$
		
		\FOR {$t = 0: T_0$}
		\STATE Observe $x_t$
		\STATE Play $u_t=K_0 x_t+\eta_t$
		\STATE observe $x_{t+1}$ 
		\STATE using $u_t$ and $x_t$ form $z_t$ and save $(z_t, x_{t+1})$ and update the ellipsoid according to the procedure given in \ref{skjs}, i.e., (\ref{eq:thetaHatWarmup}-\ref{radius_centralEl_warmUp})
		\ENDFOR
		\STATE \textbf{Output:} $\Theta_0$  the center of constructed ellipsoid
		\end{algorithmic}
	\end{algorithm}
	
	\begin{algorithm} 
	\caption{Learn and Control while Switching Algorithm(LCSA)} \label{alg:OSL311}
	\begin{algorithmic}[1]
        
		\STATE \textbf{Inputs:} $\mathcal{I}=[i_{\tau_1}\; i_{\tau_2}\;... i_{\tau_{ns}}],$ $\Theta_0,$ $\lambda=8\kappa_*^8 \sqrt{n+m} \frac{4\bar{\mu}\bar{\nu}}{\alpha_0\sigma^2}$
		\STATE {Initialize central ellipsoid by (\ref{eq:initialCentaEl})} 
        \STATE Set $\tau_1=0$
		
		\FOR {$k=1: ns$}

          \STATE Project central ellipsoid $\mathcal{C}_{\tau_k}(\delta)$ into space of mode $i_{\tau_k}$ and obtain

          \begin{align*}
		\operatorname{Tr}((\Theta^{i_{\tau_k}}-\hat{\Theta}_{\tau_k}^{i_{\tau_k}})^\top V^{pi_{\tau_k}}_{[0\;\tau_k]}(\Theta^{i_{\tau_k}}-\hat{\Theta}_{\tau_k}^{i_{\tau_k}}))\leq r_{\tau_k}
		\end{align*}

         \STATE Set $V^{i_{\tau_k}}_{[0,\tau_k]}=V^{pi_{\tau_k}}_{[0\;\tau_k]}$,  and $r^i=r_{\tau_k}$
         
        \FOR {$\tau_k\leq t < \tau_{k+1}$}
         \STATE receive state $x_t$
         \STATE Set $\mu_{i_{\tau_k}}=r^{i_{\tau_k}}+(1+r^{i_{\tau_k}})\vartheta \|V^{i_{\tau_k}}_{[0\;t]}\|^{0.5}$
         \STATE Set $\tau=\tau_k$

		\IF {$\det(V^i_{[0,\;t]})> 2 \det(V^i_{[0,\;\tau]})$ or $t=\tau_k$} 
		
		\STATE start new episode and set $\tau=t$
		
		\STATE \textbf{ estimate system parameters:} Compute the least square estimate $\hat{\Theta }^{i_{\tau_k}}_t$ by solving  (\ref{eq:obswpe})
		\STATE \textbf{Compute policy:} Solve the dual SDP program (\ref{eq:RelaxedSDP}) for $\Sigma(\Theta^{i_{\tau_k}}_t,Q,R^{i_{\tau_k}})$ 
		\STATE Set 
  
  $\bar{K}_{i_{\tau_k}}=\Sigma_{ux}(\Theta^{i_{\tau_k}}_t,Q,R^{i_{\tau_k}})\Sigma^{-1}_{xx}(\Theta^{i_{\tau_k}}_t,Q,R^{i_{\tau_k}})$
		\ELSE
		\STATE Set
  
       $\bar{K}_{i_{\tau_k}}=K(\Theta^{i_{\tau_k}}_{\tau},Q,R^{i_{\tau_k}}),$
		\ENDIF 
		
		\STATE Apply $u^{i_{\tau_k}}_t=\bar{K}_{i_{\tau_k}}x_t$ and observe new state $x_{t+1}$.
		
		\STATE Save $(z^{i_{\tau_k}}_t,x_{t+1}) $ into database

		\STATE Update $V^{i_{\tau_k}}_{[0,\; t+1]}=V^{i_{\tau_k}}_{[0,\; t]}+z^{i_{\tau_k}}_t{z^{i_{\tau_k}}_t}^\top$
		
		\STATE \textbf{ compute}\\
		$r^{i_{\tau_k}}=\bigg( \sigma_{\omega} \sqrt{2n \log \frac{n\det(V^{i_{\tau_k}}_{[0\;t+1]}) }{\delta \det(V^{i_{\tau_k}}_{[0,\tau_k]})}}+\sqrt{r_{\tau_k}}\bigg)^2, $ 
		
		\STATE Apply augmentation technique (Section \ref{augument}) to update central ellipsoid.
		\ENDFOR
        \ENDFOR
	\end{algorithmic}
\end{algorithm}
\subsection{Main steps of Algorithms}\label{AA} 

For the sake of the reader more interested in implementation than in the algorithms' theoretical underpinnings, we start by surveying the main steps of Algorithm 1 and 2. We are mostly concerned with explaining the computations and algebraic manipulations performed in each step, and provide details about the rationale for and properties of these steps under the umbrella of performance analysis in the next section. Proofs of most of these results are relegated to the appendix.

\subsubsection{System Identification}
The system identification step aims at building confidence set around the system's true but unknown parameters. This step plays a major role in both Algorithm 1 and Algorithm 2.

\paragraph{Confidence Set construction of warm-up phase} \label{skjs}
In Algorithm 1, using data obtained by executing control input $u_t=K_0 x_t+\eta_t$ (which consists of linear strongly stabilizing term and extra exploratory noise $\eta_t\sim \mathcal{N}(0, 2\sigma_{\omega}^2\kappa_0^2 I)$), a confidence set is constructed around the true parameters of the system by following the steps of Section \ref{sysIdentnai}. For warm-up phase we let $\Theta_0=0$ in (\ref{eq:LSE_pr}) which results in
\begin{align}
 \hat{\Theta}_{t} &=\operatorname*{argmin}_\Theta e(\Theta)=V_{[0,\;t]}^{-1}Z_t^\top X_t. \label{eq:thetaHatWarmup}
 \end{align}
where the covariance matrix $V_{[0,\;t]}$ is constructed similar to Section \ref{sysIdentnai}. Recalling $\|\Theta_*\|_*\leq \vartheta$ by assumption, a high probability confidence set around true but unknown parameters of system is constructed as
 \begin{align}
\nonumber \mathcal{C}_t(\delta):=\big\{&\Theta \in \mathbb{R}^{(n+m)\times n}|\\
&\operatorname{Tr}\big((\Theta-\hat{\Theta}_t)^\top V_{[0,\;t]}(\Theta-\hat{\Theta}_t)\big) &\leq r_t\big\} \label{eq:confSet1_warmup}
\end{align}
where 
\begin{align}
    r_t=\bigg( \sigma_{\omega} \sqrt{2n \log\frac{n\det(V_{[0,\;t]}) }{\delta \det(\lambda I)}}+\sqrt{\lambda} \vartheta \bigg)^2. \label{radius_centralEl_warmUp}
\end{align}
For warm-up phase we set $\lambda=\sigma_{\omega}^2\vartheta^{-1}$.

{Deployment of Algorithm 1 endures until the parameter estimates reside in $(\kappa_0,\gamma_0)$-stabilizing neighborhood  $\mathcal{N}{(\kappa_0, \gamma_0)}(\Theta*)$. As per Definition \ref{Dref:neigborhood}, the first stabilizing controller $K(\Theta_0, Q, R)$ designed by Algorithm 2, where $\Theta_0=\hat{\Theta}_{T_0}$ has been substantiated to yield an average expected cost that is lower than that of $K_0$ when applied in mode 1. This, in essence, ensures that Algorithm 2 commences with a feedback gain as effective as the initial stabilizing feedback gain $K_0$. Additionally, to secure joint stabilization and effectively manage regret, it is imperative to ensure the proximity of the estimate $\Theta_0$. This duration is determined by the theorem presented below.} 


{
\begin{theorem}\label{thm:WarmUp_Duration}
For a given $(\kappa_0, \gamma_0)$-strongly stabilizing input $K_0$, there exist explicit constants $C_0$ and $\epsilon_0=poly(\alpha_0^{-1},\alpha_1,\vartheta, \nu, n, m)$ such that the following holds:
Let $T_0$ be the smallest time that satisfies:
\begin{align}
\nonumber
\bigg( &\sigma_{\omega} \sqrt{2n \left(\log \frac{n}{\delta} + \log \left(1+\frac{300\sigma_{\omega}^2\kappa_0^4}{\gamma_0^2}(n+\vartheta^2\kappa^2_0)\log \frac{T_0}{\delta}\right)\right)} \\
&+\sqrt{\lambda} \vartheta \bigg)^2 \frac{80}{T_0 \sigma^2_{\omega}} \leq \min \bigg\{\kappa^2_0\frac{\alpha_0 \sigma^2_{\omega}}{C_0}, \epsilon^2_0, \frac{1}{\lambda}\bigg\}:=\epsilon^2. \label{eq:Durarion}
\end{align}
then that $K(\hat{\Theta}_{T_0}, Q, R)$, obtained by Algorithm 2 applied to the system $\Theta_*$, has an average expected cost $J(\hat{\Theta}_{T_0}, Q, R)$ that is lower than that of the policy $K_0$. Furthermore, requiring $\epsilon^2 \leq \frac{1}{{\lambda}}$ is a sufficient condition for $(\kappa, \gamma)-$strong sequential stability of the closed-loop system under the implementation of the policy produced by Algorithm 2.
\end{theorem}}
\vspace{0.5 cm}

{The regret associated with the implementation of the warm-up algorithm is of the order $\mathcal{O}(\sqrt{T})$. To substantiate this assertion,  let $T$ denote the specified horizon for implementing Algorithm \ref{alg:OSL311}. By definition, we have $\epsilon^2\leq 1/\lambda$, where $\lambda= 8\kappa_*^8 \sqrt{n+m}\frac{4\bar{\mu}\bar{\nu}}{\alpha_0\sigma^2}$, as the condition for strong sequential stability (see proof of Lemmas \ref{Stability_lemma18} and \ref{upper_ and lower bound of P}). This, along with the definition of $\bar{\mu}$, indeed indicates that $\lambda=\mathcal{O}(\sqrt{T})$. Now, based on equation (\ref{eq:Durarion}), we deduce that $T_0=\mathcal{O}(\sqrt{T})$. Given Proposition \ref{prop:linRegr}, which establishes that the regret associated with a fixed strongly stabilizing policy $K_0$ over $T_0$ rounds scales as $\mathcal{O}(T_0)$, we can conclude that within $T$ rounds, the warm-up phase incurs a regret of $\mathcal{O}(\sqrt{T}).$}

\paragraph{Confidence Set construction of LCSA algorithm}

Having been run for the mode 1, the warm-up phase outputs $\Theta_0$ as an initial estimate using which the central ellipsoid as an input to Algorithm 2, is constructed as follows
\begin{align}
\nonumber\mathcal{C}_{\tau_1}(\delta)=\{&\Theta\in \mathbb{R}^{(n+m)\times n}|\operatorname{Tr}((\Theta-{\Theta}_{0})^\top V_{0}(\Theta-{\Theta}_{0}))\leq r_{0}\}\\
\label{eq:initialCentaEl}
\end{align}

where $r_{0}=\lambda \epsilon$, $ V_{0}=\lambda I$ where $\lambda=8\kappa_*^8 \sqrt{n+m} \frac{4\bar{\mu}\bar{\nu}}{\alpha_0\sigma^2}$ and $r_0=\lambda \epsilon^2$. {The initial confidence ellipsoid (\ref{eq:initialCentaEl}) is, in fact, equivalent to $\|\Theta-\Theta_{0}\|^2\leq r_0$.} We apply the procedure of Section \ref{sysIdentnai} together with augmentation technique to update the central ellipsoid. 

Suppose that at time $\tau_q$, the system switches to mode $i$. The associated inherited confidence ellipsoid for mode $i$ can be derived through the projection of the central ellipsoid, as discussed in Section \ref{Sec_Proj}, and can be expressed as follows:
\begin{align}
\nonumber\mathcal{C}^{pi}_{\tau_q}(\delta):=\{&\Theta^j\in \mathbb{R}^{(n+m_i)\times n}|\\
&\operatorname{Tr} \big((\Theta^i-\hat{\Theta}^{pi}_{\tau_q})^\top V^{pi}_{[0\; \tau_q]} (\Theta^i-\hat{\Theta}^{pi}_{\tau_q})\big) &\leq r_{\tau_q}\}. \label{eq:InheritedInfo}
\end{align}
The following theorem provides a confidence bound on the error of parameter estimation when there is an inherited confidence set.

\begin{theorem}\label{thm:Conficence_Set_Attacked}
Let the linear model (\ref{eq:dynam_by_theta}) satisfy Assumption \ref{Assumption 1} with a known $\sigma_{\omega}$. Assume the system at an arbitrary time $\tau_{q}$ switches to the mode $i$ and runs for the time period $t\in[\tau_q,\;\; \tau_{q+1}]$ (i.e. $q+1$'s epoch). Let the inherited information, which is a projection of central ellipsoid, be given by (\ref{eq:InheritedInfo}). 
 Then with probability at least $1-\delta$ we have
\begin{align}
\operatorname{Tr}((\Theta^{i}-\hat{\Theta}^i_t)^\top V^i_{[0\;\; t]}(\Theta^{i}-\hat{\Theta}^i_t))\leq r^{i}\label{eq:confidence_set_noInfot}
\end{align}
where
 \begin{align}
r^i &=\bigg( \sigma_{\omega} \sqrt{2n \log \frac{n\det(V^i_{[0\;\; t]}) }{\delta \det(V_{[0\; \tau_{q}]}^{pi})}}+\sqrt{r_{\tau_q}}\bigg)^2, \label{eq:not_normalized_radius} 
\end{align}
 \begin{align}
  \hat{\Theta}^{i}_t &={V^i}^{-1}_{[0,\;t]}\big({Z}^TX+V^{pi}_{[0,\;\; \tau_q]} \hat{\Theta}^{pi}_{\tau_q} \big), \label{eq:obswpe}
  \end{align}
and
  \begin{align}
\nonumber&V_{[0 \; t]}^{i}=V_{[0\; \tau_{q}]}^{pi}+V^i_{[\tau_q\; t]},\quad
V^i_{[\tau_q\; t]}=\sum_{t=\tau_{q}}^{t} z^i_t{z^i_t}^\top,\\
&Z^\top X=\sum _{s=\tau_q}^{t}z^i_s{x^i_s}^\top \label{computation}
\end{align} 
\end{theorem}
The second term in the right hand side of (\ref{eq:not_normalized_radius}) is related to the inherited information from the central ellipsoid while the first term is the contribution of learning at the running epoch. {Similarly, in the definition of $V^i_{[0,; t]}$ provided by (\ref{computation}), the term $V_{[0; \tau_{q}]}^{pi}$ is derived from inherited information, while $V^i_{[\tau_q; t]}$ represents the contribution of learning during the current running epoch.}

\subsection{Control design} The relaxed SDP problem which accounts for the uncertainity in parameter estimates is used for control design. Given the parameter estimate $\hat{\Theta}^i_t$ and the covariance matrix $V^i_{[0,\;t]}$ obtained by (\ref{eq:obswpe}) and (\ref{computation}), the relaxed SDP is formulated as follows:
\begin{align}
\begin{array}{rrclcl}
\displaystyle  \min & \multicolumn{1}{l}{\begin{pmatrix}
	Q & 0 \\
	0 & R^i
	\end{pmatrix}\bullet \Sigma}\\
\textrm{s.t.} & \Sigma_{xx}\succeq {\hat{\Theta^i}_t}^\top \Sigma\hat{\Theta}_t^i+W-\mu_i\Sigma\bullet {{V^i}^{-1}_{[0,\;t]}}I,\\
&\Sigma\succ 0. 
\end{array}\label{eq:RelaxedSDP}
\end{align}

where $\mu_i$ is any scalar satisfying  $\mu_i\geq r^{i}+(1+r^{i})\vartheta \|V_{[0,\;t]}^i\|^{1/2}$. And we define $\bar{\mu}\geq r^{i}+(1+r^{i})\vartheta \|V_{[0,\;T]}^i\|^{1/2}$.
Denoting $\Sigma (\hat{\Theta}_t, Q, R^i)$ as the optimal solution of (\ref{eq:RelaxedSDP}), the controller extracted from solving the relaxed SDP is deterministic and linear in state ($u^i=K(\hat{\Theta}^i_t, Q, R^i)x$) where 
\begin{align*}
K(\hat{\Theta}^i_t, Q, R^i)=\Sigma_{ux}(\hat{\Theta}^i_t, Q, R^i){\Sigma^{-1}_{xx}(\hat{\Theta}^i_t, Q, R^i)}.
\end{align*} 
The designed controller is strongly stabilizing in the sense of Definition 1 and the sequence of feedback gains is strongly sequentially stabilizing, as affirmed by Theorem \ref{Stability_thm17}. The detailed analysis of the relaxation of the primal SDP formulation, incorporating the constructed confidence ellipsoid to account for parameter uncertainty, is provided in Appendix \ref{sec:relaxing SDP} which is adapted from \cite{cohen2019learning} with modifications as the result of not normalizing the confidence ellipsoid.

Recurrent update of confidence sets may worsen the performance, so a criterion is needed to prevent frequent policy switches in an actuation epoch. As such, Algorithm \ref{alg:OSL311} updates the estimate and policy when the condition $\det(V^j_t) > 2\det(V^j_{\tau})$ is fulfilled where $\tau$ is the last time the policy got updated ($\tau$ is solely used in Algorithm 2 and should not be mixed with switch times in actuating modes; $\tau_i'$s). This is a standard step in OFU-based algorithms (\cite{abbasi2011regret}, \cite{chekan2021joint}, and \cite{cohen2019learning}).

\subsection{Projection Step}
In Algorithm 2 we apply the projection technique described in Section \ref{Sec_Proj} to transfer the information contained in the central ellipsoid to the mode that the system just switched to. The central ellipsoid is consistently updated in Algorithm 2 using the augmentation technique regardless of the active actuating mode. 

\section{Performance Analysis}
\label{sec:perf}
\subsection{Stability Analysis}

In this section, we first proceed to restate the results of \cite{cohen2019learning} showing that the sequence of feedback gains obtained by solving primal relaxed-SDP (\ref{eq:RelaxedSDP}) is strongly sequentially stabilizing. We then define a Minimum Average Dwell Time (MADT), denoted as $\tau_{MAD}$, which ensures the boundedness of the state norm in a switching scenario by mandating that the system remains in each mode for at least this duration. In the TTP case, this property would need to hold for the given sequence, i.e., $\tau_{j+1}-\tau_j\geq \tau_{MAD}$, for $j\in\{1,...,ns-1\}$. For MTD and FTC cases, on the other hand, the MADT constraints the times of switch.

Solving relaxed SDP (\ref{eq:RelaxedSDP}) is the crucial step of the algorithm, however when it comes to regret bound and stability analysis, the dual program 
\begin{align}
\begin{array}{rrclcl}
\displaystyle \max & \multicolumn{1}{l}{P\bullet W}\\
\textrm{s.t.} & \begin{pmatrix}
Q-P & 0 \\
0 & R^i
\end{pmatrix}+ \hat{\Theta}_t^i P {\hat{\Theta^i}_t}^\top \succeq \mu_i \|P\|_*{{V^i}^{-1}_{[0,\;t]}}\\
&P\succeq 0 
\end{array}\label{eq:RedSDP_DUAL}
\end{align}
and its optimal solution $P(\hat{\Theta}^i_t, Q, R^i)$ play a pivotal role. 

Recalling the sequential strong stability definition, the following theorem adapted from \cite{cohen2019learning} shows that the sequence of policies obtained by solving relaxed SDP (\ref{eq:RelaxedSDP}) for system actuating in mode $i$ are strongly sequentially stabilizing for the parameters $\kappa_i$ and $\gamma_i$ to be defined.

\begin{theorem} \label{Stability_thm17}
Let $P(\hat{\Theta}^i_t, Q, R^i)$'s for $t\in\{\tau_q, \tau_q+1,...\}$ be solutions of the dual program (\ref{eq:RedSDP_DUAL}) associated with an epoch occurring in time interval $[\tau_q\; \tau_{q+1}]$  and the mode $i$ with corresponding $\Theta^i_t$ and $V^i_t$. Let $\kappa_i=\sqrt{2\nu_i /\alpha_0\sigma^2}$ and $\gamma_i=1/2\kappa_i^2$ with $\mu_i\geq r^{i}+(1+r^{i})\vartheta ||V^i||^{1/2}$, then the sequence of policies $K(\hat{\Theta}^i_t, Q, R^i)$'s (associated with $P(\hat{\Theta}^i_t, Q, R^i)$'s) are $(\kappa_i,\gamma_i)-$ sequentially strongly stabilizing {with probability at least $1-\delta$.}
\end{theorem}
\begin{proof} 
This theorem is a direct results of Lemmas \ref{Stability_lemma18} and \ref{upper_ and lower bound of P} provided and proved in Appendix.
\end{proof}
Recalling Definition 1, a sequentially stabilizing policy keeps the states of system bounded in an epoch. The following lemma gives this upper bound in terms of the stabilizing policy parameters $\kappa_i$ and $\gamma_i$.

\begin{lemma}\label{thm:res_sequential_stability_resp}
For the system evolving in an arbitrary actuating mode $i$ in time interval of $[\tau_q,\tau_{q+1}]$ by applying sequentially stabilizing controllers for any $t\in [\tau_q,\tau_{q+1}]$ the state is upper-bounded by
	\begin{align}
\|x_t\|\leq \kappa_i e^{-\gamma_i (t-(\tau_p+1))/2}||x_{\tau_q}||+\frac{20\kappa_i}{\gamma_i} \sigma_{\omega}\sqrt{n\log\frac{t-\tau_q}{\delta}} 
\end{align}
\end{lemma}
\begin{proof}
The proof is straight forward but for the sake of completeness we provide it here. Applying the sequentially stabilizing policy implies (\ref{eq:rsp_seq_st}), then using Hanson-Wright concentration inequality with probability at least $1-\delta/2(t-\tau_q)$ yields
\begin{align}
\max\limits_{\tau_q\leq s\leq t }||w_s||\leq 10\sigma_{\omega}\sqrt{n\log\frac{t-\tau_q}{\delta}}, \label{eq:upp_bound_nois}
\end{align}
which completes the proof.
\end{proof}

{Having stable subsystems is a necessary but not a sufficient condition to ensure the stability of a switched system. The switches must occur at a sufficiently slow rate on average to guarantee stability. Therefore, we need to impose a requirement that the switches occur at a sufficiently slow rate. The following theorem establishes a lower bound on the time elapsed between two consecutive switches, referred to as the "dwell time," which serves as a sufficient condition for ensuring boundedness.}

The following theorem establishes a lower bound on the average dwell time.



 {\begin{theorem}\label{thm:minimum_average_dwell_corrected} 
Consider the switch sequence $\mathcal{I} = [i_{\tau_1}, i_{\tau_2}, \ldots, i_{\tau_{ns}}]$, and let the system (\ref{eq:dynam_by_theta1}) be controlled by $(\kappa_{i_{\tau_j}}, \gamma_{i_{\tau_j}})$-strongly sequentially stabilizing controllers for any mode $i_{\tau_j} \in \mathcal{I}$. Then, for a user-defined $0 < \mathcal{X} < 1$ by requiring $\tau_{j+1} - \tau_j \geq \tau_{\text{MAD}}$, where
\begin{align}
    \tau_{\text{MAD}} := \lceil \frac{\ln \kappa^* - \ln(1-\chi)}{-\ln(1-\gamma^*)} \rceil
    \label{eq:tau_a_slow_cor} 
\end{align}
in which
\begin{align}
    \kappa^* = \max_{i \in \{1, \ldots, |\mathcal{B}^*|\}} \kappa_i, \quad \text{and} \quad \gamma^* = \frac{1}{2\kappa^{*}},
\end{align}
the state norm will be bounded as follows:
\begin{align}
    \|x_t\| \leq \kappa^* (1 - \mathcal{X})^t\|x_0\| + \frac{\bar{L}}{\mathcal{X}} \label{eq:upstate}
\end{align}
where
\begin{align*}
    \bar{L} := \frac{20\kappa^*}{\gamma^*}\sigma_{\omega}\sqrt{n\log\frac{t}{\delta}}. 
\end{align*}  
\end{theorem}}

 \subsection {Regret Bound Analysis}

Taking into account the fact that any confidence ellipsoid $\mathcal{C}_t^i(\delta)$ can be obtained through projecting the central ellipsoid $\mathcal{C}_t(\delta)$, we can define the following unified “good event”, using the denomination of \cite{abbasi2011regret}

\begin{align}
\mathcal{E}_t=\{&\operatorname{Tr}\big(\Delta^T_sV_{[0,\; s]}\Delta_s\big)\leq r_s,\\
\nonumber &\|z_s\|^2\leq 4 
	{\kappa^*}^2(1-\chi)^{2s}\|x_0\|^2+\frac{\bar{\Upsilon}}{2{\gamma^*}^2},\;\forall s=1,...,t \}\label{eq:GoodEven_state_unat} 
\end{align}
where
	\begin{align*}
	\bar{\Upsilon}=\frac{3200{\kappa^*}^6\sigma^2n}{\chi^2}\log \frac{T}{\delta}.
	\end{align*}
 
The upper bound for $\|z_t\|^2$ is obtained from (\label{eq:upstate}), as derived in the proof of Lemma \ref{define_parBeta} in Appendix \ref{BB}. The event $\mathcal{E}_t$ holds for any sequence of switches and is used for carrying out regret bound analysis. 

By an appropriate decomposition of regret on the event $\mathcal{E}_t$, for an arbitrary switch sequence $\mathcal{I}=[i_{\tau_1}\; i_{\tau_2}\;...i_{\tau_{ns}}]$, the following upper-bound for regret is obtained (adapted from \cite{cohen2019learning}):
\begin{align}
R_T \leq R_1+R_2+R_3+R_4 
\end{align}
where:
	\begin{align}
	\nonumber R_1=\sum _{t=0}^{T-1} \big(&x_{t}^T P(\hat{\Theta}^{\sigma(t)}, Q, R^{\sigma(t)})x_{t}-\\
	&x_{t+1}^T P(\hat{\Theta}^{\sigma(t)}, Q, R^{\sigma(t)})x_{t+1}\big)1_{\mathcal{E}_t},\label{eq:R1} 
	\end{align}
	\vspace{-0.5 cm}
	\begin{align}
   \nonumber  R_2=\sum _{t=0}^{T} &w^T_tP(\hat{\Theta}^{\sigma(t)}, Q, R^{\sigma(t)})\big(A_*+\\
    &B^{\sigma(t)}_*K(\hat{\Theta}^{\sigma(t)}, Q, R^{i(t)})\big)x_t 1_{\mathcal{E}_t},\label{eq:R2}
    \end{align}
	\vspace{-0.5 cm}
	\begin{align}
     \nonumber R_3 =\sum _{t=0}^{T} \big(& w^T_tP(\hat{\Theta}^{\sigma(t)}, Q, R^{\sigma(t)})w_t-\\
     &\sigma^2\|P(\hat{\Theta}^{\sigma(t)}, Q, R^{i(t)})\|_*\big)1_{\mathcal{E}_t},\label{eq:R3}
     \end{align}
	\vspace{-0.5 cm}
	\begin{align}
     R_4 =\frac{4\nu}{\sigma_{\omega}^2}\sum _{t=0}^{T} \mu_{\sigma(t)}\big(z^T_t{v^{\sigma(t)}_{[0,\;t]}}^{-1}_tz_t\big)1_{\mathcal{E}_t},\label{eq:R4}  
	\end{align}
{in which $\sigma(t) = i_{\tau_j}$ for $\tau_j \leq t \leq \tau_{j+1}$ and $1_{\mathcal{E}_t}$ is the indicator function.} By bounding the above terms separately, the overall imposed regret, not-counting the warm-up regret, can be obtained. The following theorem summarizes the result.

\begin{theorem}\label{thm:RegretBound}
	The regret of projection-equipped algorithm in horizon $T$ and with $ns$ number of switches is  $\mathcal{\bar{O}}\big(\sqrt{T}\big)+\mathcal{O}\big((ns)\sqrt{T}\big)$ {with probability at least $1-\delta$.} Moreover, this strategy is more efficient than repeatedly and naively applying the basic OFU algorithm.
\end{theorem}

\section {Modifications to Address the MTD and FTC scenarios} \label{Sec:Extension}

Algorithm 2, with minor adjustments, can also address the FTC and MTD problems where switch sequence information (actuating modes and switch times) are unavailable. This section proposes possible directions to adjust the strategy to the mentioned applications. 

For the FTC application, leaving aside switch time, which is supposedly determined by an external unit (e.g., a detector), the remaining challenge is determining the next actuating mode to switch to. From the set of available modes $\mathcal{M}$, the mode that guarantees to achieve lower regret seems a rational choice. Proofs of Lemmas \ref{Thm.BoundR1} and \ref{thm:boundR4}, in the regret bound analysis section, justify that switching to a mode whose covariance matrix has the highest determinant value guarantees achieving the lower regret bound. In other words, at the time $t$ the best mode which guarantees better regret is the mode $j^*:=\arg_{i\in \mathcal{M}}\max_{i}\det(V^{pi}_{[0,\;t]})$ where $V^{pi}_{[0,\;t]}$ is the projection of central ellipsoid on the space of mode $i \in \mathcal{M}$. 


In the MTD application, where the frequent switch is the central part of the strategy, in addition to determining the next actuating mode, we have to specify the switch instances so that the stability is not violated. To address the latter, we design a MADT (see Theorem \ref{thm:minimum_average_dwell_corrected}) and require the algorithm to make the switches "slow enough" by waiting on a mode for at least the MADT unit of time. Regarding the choice of the next mode, picking merely the well-learned one (by evaluating the covariance matrices) does not seem practical. The reason is that this strategy pushes the algorithm to choose the modes with more actuators (which have greater values of covariance matrix determinant) and technically ignore most of the modes, so the set of candidate modes will be smaller. In turn, this means less unpredictability and helps the attacker perform effective intrusion. To resolve this issue, one may decide to switch by optimizing the weighted sum of determinant values of covariance matrices (of the modes) and the value of some unpredictability measures. This idea basically gives a trade-off between lower regret and unpredictability. The authors in \cite{kanellopoulos2019moving} use an information entropy measure produced by a probability simplex when the system is known.

\section {Numerical Experiment}
\label{sec:numEx}

In this section, to demonstrate the effectiveness of the proposed ``learn and control while switching" algorithm compared to repeatedly applying the basic OFU algorithm, we use a third-order linear system with three actuators as of the following plant and control matrices:
\begin{align*}
A_*=\begin{pmatrix}
	10.4 & 0 & -2.7\\
	5.2 & -8.1 & 8.3\\
	0 & 0.4 & -9.0
	\end{pmatrix},
\;	B_*=\begin{pmatrix}
	-4.7 & 6.1 & -2.9\\
	-5.0 & 5.8 & 2.5\\
	2.9 & 0 & -7.2
	\end{pmatrix}
\end{align*}
which are unknown to the learner. Weighting cost matrices are chosen as follows:
 \begin{align*}
 Q=\begin{pmatrix}
 6.5 & -0.8 & -1.4\\
 -0.8 & 5.7 & 2.6\\
-1.4 & 2.6 & 25
 \end{pmatrix},\; R=\begin{pmatrix}
 40 & 10 & 16\\
 10 & 28 & 8\\
 16 & 8 & 48
 \end{pmatrix}.
 \end{align*}

Figure 2 shows the regret bound of the algorithm for the switch sequence of $\{1,2,1\}$ where the mode 2 is specified by
\begin{align*}
B_*^2=\begin{pmatrix}
	-4.7 & 6.1\\
	-5.0 & 5.8\\
	2.9 & 0
	\end{pmatrix},\;\; R^2=\begin{pmatrix}
 40 & 10\\
 10 & 28\\
 \end{pmatrix}.
\end{align*}

We let the switch happens at times $\tau_1=0$, $\tau_2=5000$, and $\tau_3=10000$ seconds. We intentionally picked this switch sequence to show the regret when the confidence set is projected to a lower-dimensional space and when the system switches back to a previously explored mode. We further assume the initial estimate $\Theta_0$ where $\|\Theta_0 -\Theta_*\|\leq \epsilon$ is given and let the naive algorithm use this estimate at time $\tau_2$ when the system switches to the second mode. We set $T=15000$, $\delta =0.5/T$, $\sigma_{\omega}=0.003$, $\nu=0.8$ and $\epsilon=0.0001$. 

At the time $\tau_2=5000$ seconds, when the switch happens, the projection-equipped algorithm outperforms the naive one thanks to the learned information during the first epoch, which is restored through the projection tool. At the time $\tau_3=10000$ seconds, the system switches back to mode 1. During this last epoch, we only see a slight improvement which is because the switch in policy, triggered by fulfillment of the condition $\det (V_{[0,t]})>\det (V_{[0,\tau]})$ (see the line 8 of algorithm 2) rarely happens as the value of $\det (V_{[0,t]})$ increases. In our case, we only have one time of policy improvement during the epoch $[\tau_3,\; T]$.

\begin{figure}[thpb]
	\centering
	\vspace{5pt}
	\includegraphics[scale=.4]{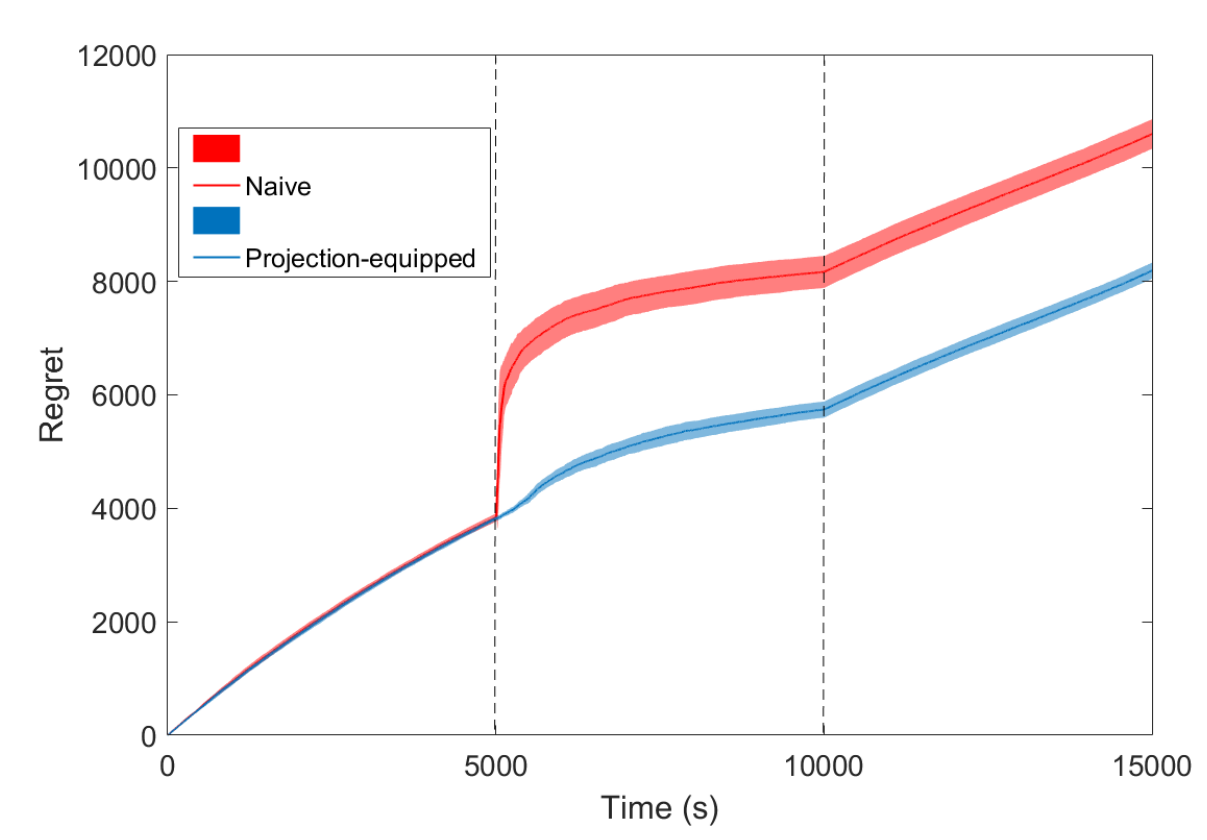}
	\vspace{-0.3 cm}
	\caption{Regret bound of the projection-based algorithm vs naive algorithm }
	\label{Regret_Bound_switched}
\end{figure}

\section{Summary and Conclusion} \label{sec:Conclusion}
In this work, we equipped the OFU algorithm with a projection tool to improve the "learn and control" task when switching between the different subsets of actuators of a system is mandated. This improvement is proved mathematically in a regret-bound sense and shown throughout the simulation compared to a case when the basic OFU algorithm is repeatedly applied. We presented this idea for a time-triggered network control system problem when the sequences of switch times and modes are upfront. We also discussed the possibilities of extending the proposed algorithm to applications such as fault-tolerant control and moving target defense problems for which there is a freedom to pick the best actuating mode when it is a time to switch. To address this, we proposed a mechanism that offers the best-actuating mode, considering how rich the learned data is. A minimum average dwell-time is also designed to address stability concerns in applications such as moving target defense that frequent switches might cause instability. Moreover, by applying recent results on the existence of a stabilizing neighborhood, we have determined an initial exploration duration such that the efficiency for the warm-up phase is guaranteed.

\bibliographystyle{IEEEtran}
\bibliography{ref}

\appendices




\section{Further Analysis}
\label{partD}
In this section, we dig further and provide proofs, rigorous analysis of the algorithms, properties of the closed-loop system, and regret bounds. 

\subsection{Technical Theorems and Lemmas}
The following lemma, adapted from \cite{abbasi2011online} gives a self-normalized bound for scalar-valued martingales. 
 
\begin{lemma} \label{Self_normalized_Bound_1}
	 Let $F_k$ be a filtration, $z_k$ be a stochastic process adapted to $F_k$ and $\omega^i_k$ (where $\omega^i_k$ is the $i-$th element of noise vector $\omega_k$) be a real-valued martingale difference, again adapted to filtration $F_k$ which satisfies the conditionally sub-Gaussianity assumption (Assumption \ref{Assumption 1}) with known constant $\sigma_{\omega}$. 
	 Consider the martingale and co-variance matrices:
	\begin{align*}
S^i_t:=\sum _{k=1}^{t} z_{k-1}\omega^i_k, \quad V_{t}=\lambda I+\frac{1}{\beta}\sum _{s=1}^{t-1}z_{t} z_{t}^T
	\end{align*}
	
then with probability of at least $1-\delta$, $0<\delta<1$ we have,
\begin{align}
\left\lVert S^i_t\right\rVert^2_{V_{t}^{-1}} \leq 2 \sigma_{\omega}^2\log \bigg(\frac{\det(V_t) }{\delta \det(\lambda I)}\bigg)
\end{align}	
\end{lemma}
\begin{proof}
The proof can be found in Appendix \ref{partF}.
\end{proof}

The following corollary generalizes the Lemma \ref{Self_normalized_Bound} for a vector-valued martingale which will be later used in Theorem \ref{thm:Conficence_Set_Attacked} to derive the confidence ellipsoid of estimates.

\begin{corollary} \label{Self_normalized_Bound}
	Under the same assumptions as in the Lemma \ref{Self_normalized_Bound}, $\beta>1$ and defining
\begin{align*}
S_t=Z^T_tW_t=\sum _{k=1}^{t} z_{k-1}\omega^T_k
\end{align*}
with probability at least $1-\delta$,
\begin{align}
\operatorname{Tr} (S^T_tV^{-1}_tS_t)&\leq 2\sigma_{\omega}^2 n\log \bigg(\frac{n \det(V_t) }{\delta\det(\lambda I)}\bigg)\label{eq:mart_vec}.
\end{align}
\end{corollary}
\begin{proof}
For proof see Appendix \ref{partF}.
\end{proof}

\begin{lemma}(Mania et al \cite{mania2019certainty}) \label{lem:mania}
    There exists explicit constants $C_0$ and $\epsilon_0=poly(\alpha_0^{-1},\alpha_1,\vartheta, \nu, n, m)$ such that $\|\Theta-\Theta_*\|\leq \epsilon$ for any $0\leq \epsilon \leq \epsilon_0$ guarantees 
    \begin{align}
        J(K(\Theta, Q, R), \Theta_*, Q,R)-J_*\leq C_0 \epsilon ^2
    \end{align}
   where  $J(K(\Theta, Q, R), \Theta_*, Q,R)$ is average expected cost of executing the controller $J(K(\Theta, Q, R)$ on the system $\Theta_*$ and $J_*$ is optimal average expected cost.
\end{lemma}

\begin{lemma}
 (Cassel et al \cite{cassel2020logarithmic}) \label{lem:asef} Suppose $J(K(\Theta, Q, R), \Theta_*, Q,R)\leq \mathcal{J}$ then $K(\Theta, Q, R)$ is $(\kappa, \gamma)$-strongly stabilizing with $\kappa=\sqrt{\frac{\mathcal{J}}{\alpha_0 \sigma^2_{\omega}}}$ and $\gamma=\frac{1}{2 \kappa^2}$.
\end{lemma}

\subsection{Confidence Construction (System Identification)}\label{sec: Confidence Set of Switched}

\textbf{Proof of Theorem \ref{thm:WarmUp_Duration}} 
 
 \begin{proof}
 Note that during the warm-up phase we execute the control input $u_t=K_0 x_t+\eta_t$ which includes a more-exploration term of$\eta_t \sim \mathcal{N}(0, 2\sigma_{\omega}^2\kappa_0^2 I)$. The state norm of system is upperbounded by strong stabilizing property of $K_0$ during warm-up phase (See \cite{cohen2019learning}). Now, let $T_0$ denote warm-up duration. The confidence ellipsoid is constructed during this phase by applying the procedure illustrated in Section \ref{sysIdentnai}. 
  \begin{align}
\nonumber \mathcal{C}_t(\delta):=\{&\Theta\in \mathbb{R}^{(n+m)\times n}|\\
\nonumber&\;Tr \big((\Theta-\hat{\Theta}_{T_0})^\top V_{[0,\;T_0]} (\Theta-\hat{\Theta}_{T_0})\big) &\leq r_{T_0}\} 
\end{align}
  where
 \begin{align}
\nonumber r_{T_0}=\bigg(&\sigma_{\omega} \sqrt{2n (\log \frac{n}{\delta} + \log (1+\frac{300\sigma_{\omega}^2\kappa_0^4}{\gamma_0^2}(n+\vartheta^2\kappa^2_0)\log \frac{T_0}{\delta}}\\
\nonumber&+\sqrt{\lambda} \vartheta \bigg)^2.
\end{align}
 Furthermore, from Theorem 20 of \cite{cohen2019learning}, we have the following lower bound for $V_{[0,\;T_0]}$
\begin{align*}
  V_{[0,\;T_0]} \succeq \frac{T_0\sigma_{\omega}^2}{80}I,  
\end{align*}
which results in 
\begin{align}
&\nonumber \big\|\hat{\Theta}-\Theta_*\big\|^2\leq\\
\nonumber &\frac{80}{T_0 \sigma^2_{\omega}}\bigg( \sigma_{\omega} \sqrt{2n (\log \frac{n}{\delta} + \log (1+\frac{300\sigma_{\omega}^2\kappa_0^4}{\gamma_0^2}(n+\vartheta^2\kappa^2_0)\log \frac{T_0}{\delta}} \\
&+\sqrt{\lambda} \vartheta \bigg)^2.
\end{align}
  Now applying the results of Lemmas \ref{lem:mania} and \ref{lem:asef} and setting $\kappa=\kappa_0$ yield
  \begin{align}
\nonumber &\frac{80}{T_0 \sigma^2_{\omega}}\bigg( \sigma_{\omega} \sqrt{2n (\log \frac{n}{\delta} + \log (1+\frac{300\sigma_{\omega}^2\kappa_0^4}{\gamma_0^2}(n+\vartheta^2\kappa^2_0)\log \frac{T_0}{\delta}}\\
&+\sqrt{\lambda} \vartheta \bigg)^2\leq \min \bigg\{\kappa^2_0\frac{ \alpha_0 \sigma^2_{\omega}}{C_0}, \epsilon^2_0 \bigg\}.
\end{align}
 When Algorithm 2 begins actuating in mode 1, given that $\Theta_0$ resides within a $(\kappa_0, \gamma_0)$-stabilizing neighborhood, solving the primal problem with the confidence bound constructed by $\Theta_0$ results in an average expected cost improvement over that of $K_0$, or at least achieves a similar performance.
 \end{proof} 

\textbf{Proof of Theorem \ref{thm:Conficence_Set_Attacked}}

\begin{proof}
The $l_2-$ least square error regularized by (\ref{eq:InheritedInfo}) is written as follows
   \begin{align}
  \nonumber e(\Theta^i)&=\operatorname{Tr} \big((\hat{\Theta}^{pi}_{\tau_q}-\Theta^i)^\top V_{[0\; \tau_{q}]}^{pi}(\hat{\Theta}^{pi}_{\tau_q}-\Theta^i)\big)+\\
  &\quad\sum _{s=\tau_q+1}^{t-1} \operatorname{Tr} \big((x^i_{s+1}-{\Theta^i}^\top z^i_{s})(x^i_{s+1}-{\Theta^i}^\top z^i_{s})^T\big). \label{eq:LSE_prii}
  \end{align}
Solving 
\begin{align}
\begin{array}{rrclcl}
\displaystyle  \min_{\Theta^i} & \multicolumn{1}{l}{e(\Theta^i)}\\
\end{array}\label{eq:minimize_findTheta}
\end{align}
  yields
 \begin{align}
  \hat{\Theta}^{i}_t &={V^i}^{-1}_{[0,\;t]}\big({Z}^TX+V^{pi}_{[0,\;\; \tau_q]} \hat{\Theta}^{pi}_{\tau_q} \big), \label{joonijoon}
  \end{align}
 where
  \begin{align}
&\nonumber V_{[0 \; t]}^{i}=V_{[0\; \tau_{q}]}^{pi}+V^i_{[\tau_q\; t]},\; 
V^i_{[\tau_q\; t]}=\sum_{t=\tau_{q}}^{t} z^i_t{z^i_t}^\top,\\
&Z^\top X=\sum _{s=\tau_q}^{t}z^i_s{x^i_s}^\top 
\end{align}

Substitute $X=Z^i \Theta^i_{*}+W$ (counterpart of (\ref{compactdyn}) for mode $i$), where $X$, $Z$, and $W$ whose rows are $x^\top_{\tau_q+1}, ...,x^\top_t$,  ${z}^\top_{\tau_q}, ...,{z}_{t-1}$, and $\omega_{\tau_q+1}^\top,...,\omega_t^\top$ into (\ref{joonijoon}),  yields, 

  \begin{align*}
    \hat{\Theta}^{i}_t &={V^i}^{-1}_{[0,\;t]}\big(Z^T(Z\Theta^i _{*}+W)+V_{[0\; \tau_{q}]}^{pi} \hat{\Theta}^{pi}_{\tau_q} \big)\\
  &={V^i}^{-1}_{[0,\;t]}\big(V_{[0 \; t]}^{i}\Theta^i_{*}+V^{pi}_{[0,\;\tau_q]} (\hat{\Theta}^{pi}_{\tau_q}-\Theta^i_{*})+
   Z^\top W \big)\\
   &=\Theta^i_{*}+{V^i}^{-1}_{[0,\;t]}\big(V^{pi}_{[0,\;\tau_q]} (\hat{\Theta}^{pi}_{\tau_q}-\Theta^i_{*})+ Z^\top W\big)
  \end{align*}
  where $Z^\top W=\sum _{s=\tau_q}^{t}z^i_s\omega^\top_s$. For an arbitrary random extended state $z^i$, one can write
  
  \begin{align*}
  {z^i}^\top\hat{\Theta}^i_i-{z^i}^\top\Theta^i_{*} =&\langle\,z^i,Z^\top W\rangle_{{V^i}^{-1}_{[0,\;t]}}+\\
  &\langle\,z^i,V^{pi}_{[0,\; \tau_q]} (\hat{\Theta}^{pi}_{\tau_q}-\Theta^i_{*})\rangle_{{V^i}^{-1}_{[0,\;t]}}
  \end{align*}
   which leads to
  \begin{align}
  \nonumber\| {z^i}^\top\hat{\Theta}^i_i-{z^i}^\top\Theta^i_{*}\| \leq \| z^i\|_{{V^i}^{-1}_{[0,\;t]}}&\Bigg (\| Z^\top W \|_{{V^i}^{-1}_{[0,\;t]}}+\\
 \nonumber &\| V^{pi}_{[0,\; \tau_q]} (\hat{\Theta}^{pi}_{\tau_q}-\Theta^i_{*}) \|_{{V^i}^{-1}_{[0,\;t]}}\Bigg)
  \end{align}
  which can  equivalently be written as:
  
    \begin{align}
  &\nonumber\|\hat{\Theta}^{i}-\Theta^i_*\|^2_{{V^i}_{[0,\;t]}} \leq\\
  &\Bigg (\|  Z^\top W \|_{{V^i}^{-1}_{[0,\;t]}} +
  \| V^{pi}_{[0,\; \tau_q]} (\hat{\Theta}^{pi}_{\tau_q}-\Theta^i_{*}) \|_{{V^i}^{-1}_{[0,\;t]}}\Bigg)^2
   \end{align}
   
  Given the fact that $V^i_{[0,\;t]}\geq V^{pi}_{[0,\; \tau_q]}$, 
    \begin{align*}
&\|\hat{\Theta}_{i}-\Theta_{i*}\|^2_{V^i_{[0,\;t]}}\leq\\
& \Bigg (\|  Z^\top W \|_{{V^i}^{-1}_{[0,\;t]}} + \sqrt{\| \hat{\Theta}^{pi}_{\tau_q}-\Theta^i_{*} \|^2_{{V^{pi}}^{-1}_{[0,\;\tau_q]}}}
 \Bigg)^2
  \end{align*}
    
holds true where the first term, $\|  Z^\top W \|_{{V^i}^{-1}_{[0,\;t]}}$, is bounded by Corollary \ref{Self_normalized_Bound} with probability at least $1-\delta$ as follows:
   \begin{align}
   \|  Z^\top W \|_{{V^i}^{-1}_{[0,\;t]}}^2 \leq 2 \sigma_{\omega}^2 \log \bigg(\frac{n\det(V^i_{[0,\; t]}) }{\delta \det(V_{[0\; \tau_{q}]}^{pi})}\bigg) \label{eq:mart_bound}
   \end{align}
and the second term is bounded by  (\ref{eq:InheritedInfo}) as follows
  \begin{align}
\| \hat{\Theta}^{pi}_{\tau_q}-\Theta^i_{*} \|_{{V^{pi}}^{-1}_{[0,\;\tau_q]}}^2\leq r_{\tau_q} \label{eq:second_bound}
  \end{align}
with probability at least $1-\delta $. 
Combining bounds (\ref{eq:mart_bound}) and (\ref{eq:second_bound}) completes the proof.
\end{proof}

The following lemma aims at obtaining an upper-bound for the right hand side of (\ref{eq:not_normalized_radius}) which will be later used in regret bound analysis.
%
 \begin{lemma}\label{thm:compact_ellips}
The radius $r^j$ of ellipsoid (\ref{eq:confidence_set_noInfot}) can be upper-bounded as
 \begin{align}
\nonumber r^{j}\leq 8\sigma_{\omega}^2n \bigg(&2\log\frac{n}{\delta}+(1+2\bar{\Upsilon})\log t+\\
&(m-1)\log(1+2\bar{\Upsilon})t\bigg)+2\epsilon \lambda=:\bar{r}. \label{eq:Upperbound_forr^m}
 	\end{align}	
where
	\begin{align*}
	\bar{\Upsilon}=\frac{3200{\kappa^*}^6\sigma^2n}{\chi^2}\log \frac{T}{\delta}.
	\end{align*}
 \end{lemma}

 \subsection{Properties of Relaxed SDP} \label{sec:relaxing SDP}
This subsection attempts to go through some essential features of the primal and dual relaxed-SDP problems (\ref{eq:RelaxedSDP}) and (\ref{eq:RedSDP_DUAL}). While the former is solved by Algorithm 2 to compute a regularizing control in the face of parameter uncertainty, the latter plays a pivotal role in the stability and regret bound analysis.

Both primal and dual problems directly follow from applying next two matrix-perturbation lemmas \ref{matix_perturbation_Bound1} and \ref{matix_perturbation_Bound} (adapted from \cite{cohen2019learning}) on the exact primal SDP and its dual problems. As the following lemmas hold for any actuating mode $i$, we deliberately dropped the superscripts $i$. 
\begin{lemma} \label{matix_perturbation_Bound1}
	Let $X$ and $\Delta$ be matrices of matching sizes and let $\Delta\Delta^T\leq r V^{-1}$ for positive definite matrix $V$. Then for any $P\geq 0$ and $\mu\geq r+(1+r)\|X\|\|V\|^{1/2}$, we have
	\begin{align*}
-\mu \|P\|_*V^{-1}\leq (X+\Delta)^\top P(X+\Delta)\leq \mu \|P\|_*V^{-1}.
	\end{align*}

\end{lemma}
\begin {proof}
The proof follows the same steps of Lemma 24 in \cite{cohen2019learning}. Furthermore, the lower bound of $\mu$ is obtained by following the same steps of \cite{cohen2019learning} exactly. However, our bound is slightly different than that of \cite{cohen2019learning} because, in the confidence set construction, we did not apply normalization.
\end{proof}

To discuss the solution property of primal problem we need following lemma adapted from \cite{cohen2019learning}.

\begin{lemma} \label{matix_perturbation_Bound}
	Let $X$ and $\Delta$ be matrices of matching sizes and let $\Delta\Delta^\top \leq r V^{-1}$ for positive definite matrix $V$. Then for any $\Sigma\geq0$ and $\mu\geq r+(1+r)\|X\|\|V\|^{1/2}$,
	\begin{align*}
    \|(X+\Delta)^\top \Sigma(X+\Delta)-X^\top \Sigma X\|\leq\mu\Sigma\bullet V^{-1}.
	\end{align*}
	
\end{lemma}
\begin{proof}
	The proof can be found \cite{cohen2019learning}.
\end{proof}

By substituting $X=\hat{\Theta}^i_t$, $\Delta=\hat{\Theta}_t^i-\Theta_*^i$, $\Sigma=\Sigma(\hat{\Theta}^i_t, Q, R^i)$,  $P=P(\hat{\Theta}^i_t, Q, R^i)$ and $V=V^i_{[0,\;t]}$ in Lemmas \ref{matix_perturbation_Bound1} and \ref{matix_perturbation_Bound}, and applying them on exact primal ((\ref{eq:SDPKhali}) with $\Theta=\Theta^i_*$) and its dual SDPs, the primal and dual relaxed-SDP problems ((\ref{eq:RelaxedSDP}) and (\ref{eq:RedSDP_DUAL})) are obtained. Moreover, it is straightforward to show $\mu_i\geq r^i+(1+r^i)\vartheta \|V^i_{[0,\;t]}\|^{1/2}$.
 
 Furthermore, it is also shown that $\Sigma(\Theta_*^i, Q, R^i)$, the solution of SDP (\ref{eq:SDPKhali}) (with $\Theta=\Theta_*^i$ and $R=R^i$), is a feasible solution of the relaxed SDP (\ref{eq:RelaxedSDP}). 
 
We can show that the primal-dual solutions of (\ref{eq:RelaxedSDP}) and (\ref{eq:RedSDP_DUAL})) respectively satisfy $\|\Sigma(\hat{\Theta}^i_t, Q, R^i)\|_*\leq J^*_i/\alpha_0$ and $\|P(\hat{\Theta}^i_t, Q, R^i)\|_*\leq J_*^i/\sigma^2$ where $J_*^i$ is the corresponding optimal average expected cost of actuating in mode $i$ (See Lemma 15 in \cite{cohen2019learning} for more details). 

The following lemma shows how to extract the deterministic linear policy from the relaxed SDP. We will later apply this lemma in the Lemma \ref{Stability_lemma18} to show that for all actuating modes, the designed controller is strongly sequentially stabilizing in an epoch (a time interval between two subsequent switches).

\begin{lemma} \label{matix_policy}
	Assume $\mu_i\geq r^{i}+(1+r^{i})\vartheta \|V^i_{[0,\;t]}\|^{1/2}$ and $V^i_{[0,\;t]}\geq (\nu_i\mu_i/\alpha_0\sigma^2)I$ (This holds by having definition $\lambda\geq 4\bar{\nu} \bar{\mu}/\alpha_0\sigma^2$ given in Lemma \ref{Stability_lemma18} with $\bar{\nu}=\max_{i}\nu_i$ and $\bar{\mu}>\mu_i$'s). Let $(\hat{A}_t, \hat{B}_t^i)^\top=\hat{\Theta}_t^i$ and $\Sigma(\hat{\Theta}^i_t, Q, R^i)$ and $P(\hat{\Theta}^i_t, Q, R^i)$, denote primal-dual solutions of the relaxed SDPs (\ref{eq:RelaxedSDP}) and (\ref{eq:RedSDP_DUAL}). Then, $\Sigma_{xx}(\hat{\Theta}^i_t, Q, R^i)$ is invertible and 
	\begin{align}
\nonumber P(\hat{\Theta}^i_t, Q, R^i)&=Q+K^\top (\hat{\Theta}^i_t, Q, R^i)R^i K(\hat{\Theta}^i_t, Q, R^i)+\\
\nonumber &(\hat{A}_t+\hat{B}_t^iK(\hat{\Theta}^i_t, Q, R^i))^\top P(\hat{\Theta}^i_t, Q, R^i)\times \\
\nonumber &(\hat{A}_t+\hat{B}_t^iK(\hat{\Theta}^i_t, Q, R^i))-\mu_i\|P(\hat{\Theta}^i_t, Q, R^i)\|_*\times \\
&\begin{pmatrix} I \\ K(\hat{\Theta}^i_t, Q, R^i) \end{pmatrix}{V^i}^{-1}_{[0,\;t]}\begin{pmatrix} I \\ K(\hat{\Theta}^i_t, Q, R^i) \end{pmatrix}^\top
\label{eq:Lemma8_res}
	\end{align}
	where $ K(\hat{\Theta}^i_t, Q, R^i)=\Sigma_{ux}(\hat{\Theta}^i_t, Q, R^i){\Sigma}^{-1}_{xx}(\hat{\Theta}^i_t, Q, R^i)$.
\end{lemma}
\begin{proof} 
	proof is in Lemma 16 of \cite{cohen2019learning}.
\end{proof}

Now we continue with carrying out the stability analysis of the system under execution of control designed by Algorithms 3 and 4. 

\subsection{Stability Analysis}\label{BB}

Proof of Theorem \ref{Stability_thm17} directly follows by Lemma \ref{Stability_lemma18}, and \ref{upper_ and lower bound of P} proved below.

Lemma \ref{Stability_lemma18} shows that the designed controller $K(\hat{\Theta}^i_t, Q, R^i)$ is $(\kappa_i,\gamma_i)-$ strongly stabilizing for the values of parameters specified in Theorem \ref{Stability_thm17}. Lemma \ref{upper_ and lower bound of P} shows $P(\hat{\Theta}^i_t, Q, R^i)$ is closed to $P(\hat{\Theta}^i_{t+1}, Q, R^i)$ and as such proves the sequentiality of stability and completes the proof of the Theorem \ref{Stability_thm17}.

\begin{lemma} \label{Stability_lemma18}
For the system which actuates in an arbitrary mode $i$, $K(\hat{\Theta}^i_t, Q, R^i)$ is $(\kappa_i,\gamma_i)-$strongly stabilizing for  $A_*+B^i_*K(\hat{\Theta}^i_t, Q, R^i)=H^i_tL^i_t{h^i}^{-1}_t$ where $H^i_t={P}^{1/2}(\hat{\Theta}^i_t, Q, R^i)$ and $\|L^i_t\|\leq1-\gamma_i$. Moreover, $(\alpha_0/2)I\leq P(\hat{\Theta}^i_t, Q, R^i)\leq(\nu_i/\sigma_{\omega}^2)I$ 
\end{lemma}
\begin{proof} 
Regardless of the similarity of the proof to that of Lemma 18 in \cite{cohen2019learning}, we provide it for the sake of completeness and justification of the value of $\lambda$, and upper-bound of $\mu_i$ mentioned in Theorem \ref{Stability_thm17}.

First we need to appropriately upper-bound $\mu_i$ which is an input of the algorithms 3 and 4.
	Given $\|V^i_t\|\leq (1+2\bar{\Upsilon})T\;\; \forall i$ (proved in Lemma \ref{define_parBeta}),we have
	\begin{align}
	 \nonumber \mu_i&=r^{i}+(1+r^{i})\vartheta\|V^i_{[0,\;t]}\|^{1/2}\\
	 \nonumber &\leq r^{i}+(1+r^{i})\vartheta \sqrt{(1+2\bar{\Upsilon})T}\\ &\leq \bar{r}+(1+\bar{r})\vartheta \sqrt{(1+2\bar{\Upsilon})T}:=\bar{\mu} \label{eq:def_of_muBar}
	\end{align}
	where $\bar{r}$ is an upper bound $r^i$ for all $i \in \{1,...,2^m\}$, which has already been provided by the Lemma \ref{thm:compact_ellips}:
	\begin{align*}
\bar{r}:&= 8\sigma^2_{\omega}n \bigg(2\log\frac{n}{\delta}+(1+2\bar{\Upsilon})\log T+\\
&(m-1)\log(1+2\bar{\Upsilon})T\bigg)+2\epsilon^2\lambda 
 	\end{align*}

	Recalling the dual formulation (\ref{eq:RedSDP_DUAL}) and the fact $\nu_i\geq J_*^i$ one can write 
	\begin{align*}
\nu_i\geq J_*^i\geq P(\hat{\Theta}^i_t, Q, R^i)\bullet W\geq \begin{pmatrix}
Q & 0 \\
0 & R^i
\end{pmatrix}\bullet W\geq \alpha_0\sigma^2 
	\end{align*}
	which clearly shows that $\nu_i/\alpha_0\sigma^2\geq1$. 
	
	Now we assume that $\lambda\geq 4\nu_i\mu_i/\alpha_0\sigma_{\omega}^2 \;\; \forall i\in \{1,...,2^m\}$ and	apply perturbation lemma, Lemma \ref{matix_perturbation_Bound} to (\ref{eq:Lemma8_res}) yields
\begin{align}
\nonumber &P(\hat{\Theta}^i_t, Q, R^i)=Q+K^\top (\hat{\Theta}^i_t, Q, R^i)R^i K(\hat{\Theta}^i_t, Q, R^i)+\\
\nonumber &(A_*+B_*^iK(\hat{\Theta}^i_t, Q, R^i))^\top P(\hat{\Theta}^i_t, Q, R^i)\times \\ \nonumber &(A_*+B_*^iK(\hat{\Theta}^i_t, Q, R^i))-2\mu_i\|P(\hat{\Theta}^i_t, Q, R^i)\|_*\times \\
&\begin{pmatrix} I \\ K(\hat{\Theta}^i_t, Q, R^i) \end{pmatrix}{V^i}^{-1}_{[0,\;t]}\begin{pmatrix} I \\ K(\hat{\Theta}^i_t, Q, R^i) \end{pmatrix}^\top.
\label{eq:Lemma8_res22}
	\end{align}

We have $\|P(\hat{\Theta}^i_t, Q, R^i)\|_*\leq\nu_i/\sigma_{\omega}^2$, and define  {$\lambda> (4\bar{\nu}\bar{\mu}/\alpha_0\sigma_{\omega}^2)I$} where $\bar{\nu}=\max_{i \in \{1,...,2^m\}} \nu_i$ and $\bar{\mu}=\max_{i \in \{1,...,2^m\}}\mu_i$ defined by (\ref{eq:def_of_muBar}). Given the fact that $V^i_{[0,\;t]}\geq \lambda I$ it yields	
\begin{align}\label{eq22222}
\nonumber \mu_i\|P(\hat{\Theta}^i_t, Q, R^i)\|_*{V^i}_{[0,\;t]}^{-1}&\leq \mu_i*\frac{\nu_i}{\sigma_{\omega}^2}\frac{\alpha_0\sigma_{\omega}^2}{4\bar{\nu}\bar{\mu}}I\\
&\leq  \mu_i*\frac{\nu_i}{\sigma_{\omega}^2}\frac{\alpha_0\sigma_{\omega}^2}{4\nu_i\mu_i}I=\frac{\alpha_0}{4}I 
\end{align}
plugging (\ref{eq22222}) together with the assumptions $Q, R^i\geq\alpha_0I$ into (\ref{eq:Lemma8_res22}) yields,
\begin{align}
\nonumber P(\hat{\Theta}^i_t, Q, R^i)& \geq \frac{\alpha_0}{2} I+\frac{\alpha_0}{2}K^\top (\hat{\Theta}^i_t, Q, R^i)K(\hat{\Theta}^i_t, Q, R^i)+\\
\nonumber &(A_*+B^i_*K(\hat{\Theta}^i_t, Q, R^i))^\top P(\hat{\Theta}^i_t, Q, R^i)\times\\
&(A_*+B^i_*K(\hat{\Theta}^i_t, Q, R^i)). \label{eq:inequality_of_P}
\end{align} 
Defining $\kappa_i=\sqrt{\frac{2\nu_i}{\alpha_0\sigma_{\omega}^2}}$ and using the fact that $\|P(\hat{\Theta}^i_t, Q, R^i)\|\leq\frac{\nu_i}{\sigma_{\omega}^2}$ (which implies $\kappa_i^{-2} P(\hat{\Theta}^i_t, Q, R^i)\leq \frac{\nu_i\kappa_i^{-2}}{\sigma_{\omega}^{2}}I$),

it holds true that
\begin{align}
P(\hat{\Theta}^i_t, Q, R^i)-\frac{1}{2}\alpha_0I\leq(1-\kappa_i^{-2}) P(\hat{\Theta}^i_t, Q, R^i) \label{eq:kappa_in_inequality}
\end{align}
From (\ref{eq:inequality_of_P}) and (\ref{eq:kappa_in_inequality}) one has:
	\begin{align}
\nonumber &{P}^{-1/2}(\hat{\Theta}^i_t, Q, R^i)(A_*+B^i_*K(\hat{\Theta}^i_t, Q, R^i))^\top P(\hat{\Theta}^i_t, Q, R^i)\times\\
&(A_*+B^i_*K(\hat{\Theta}^i_t, Q, R^i)){P}^{-1/2}(\hat{\Theta}^i_t, Q, R^i)\leq (1-\kappa_i^{-2})I. \label{eq:uppBound_L}
\end{align}
Denoting $H_t^i={P}^{1/2}(\hat{\Theta}^i_t, Q, R^i)$ and $L_t^i={P}^{-1/2}(\hat{\Theta}^i_t, Q, R^i)(A_*+B^i_*K(\hat{\Theta}^i_t, Q, R^i)){P}^{1/2}(\hat{\Theta}^i_t, Q, R^i)$,  (\ref{eq:uppBound_L}) leads $\|L_i^t\|\leq\sqrt{1-\kappa_i^{-2}}\leq 1-1/2\kappa_i^{-2}$. 

Furthermore, (\ref{eq:inequality_of_P}) results in $P(\hat{\Theta}^i_t, Q, R^i)\geq \frac{\alpha_0}{2}{K}^\top (\hat{\Theta}^i_t, Q, R^i)K(\hat{\Theta}^i_t, Q, R^i)$ which together with $\|P(\hat{\Theta}^i_t, Q, R^i)\|\leq\frac{\nu_i}{\sigma_{\omega}^2}$  imply $\|K(\hat{\Theta}^i_t, Q, R^i)\|\leq \kappa_i$.

This automatically yields $\|H^i_t\|\leq\sqrt{\nu_i/\sigma_{\omega}^2}$ and $\|{H^i_t}^{-1}\|\leq\sqrt{\frac{2}{\alpha_0}}$.
So, our proof is complete and $K(\hat{\Theta}^i_t, Q, R^i)$'s are $(\kappa_i,\gamma_i)-$ strongly stabilizing.
\end{proof}
Having established the strong stability by by Lemma \ref{Stability_lemma18}, Lemma \ref{upper_ and lower bound of P} shows $P(\hat{\Theta}^i_t, Q, R^i)$ is closed to $P(\hat{\Theta}^i_{t+1}, Q, R^i)$ and as such proves the sequentiality of stability and completes the proof of the Theorem \ref{Stability_thm17} (See \cite{cohen2019learning} for more detail).

{\begin{lemma} \label{upper_ and lower bound of P}
For any mode $i\in \mathcal{M}$ by choosing the regularization parameter $\lambda\geq 8 \kappa_*^8 \sqrt{n+m} \frac{4 \bar{\nu} \bar{\mu}}{\alpha_0^i \sigma_{\omega}^2}$
 we have
	\begin{align}
     \nonumber P(\hat{\Theta}^i_t, Q, R^i)\leq  P(\Theta^i_*, Q, R^i)\leq  P(\hat{\Theta}^i_{t+1}, Q, R^i)+\frac{\alpha_0^i \gamma_i}{2}I
    \end{align}
\end{lemma}}
\begin{proof}
    Proof can be found in Appendix \ref{partF}.
\end{proof}


 \textbf{Proof of Theorem \ref{thm:minimum_average_dwell_corrected} }
 
 \begin{proof}
 From the Lemma \ref{thm:res_sequential_stability_resp}, the state norm of system actuating in mode $i$ and in an arbitrary actuation epoch happening in interval $[\tau_q\;\tau_{q+1}]$ can be upper bounded as follows:
 
 \begin{align}
 \|x_t\|\leq \kappa_i (1-\gamma_i/2)^t||x_{\tau_q}||+\frac{2\kappa_i}{\gamma_i} \max\limits_{\tau_q\leq t\leq \tau_{q+1} }\|w_t\| \label{eq:OneeShotmodeepoch}
 \end{align}
 where $x_{\tau_q}$ is an initial state value of the actuation epoch. Having $\kappa^*=\max_{i\in \{1,..., |\mathcal{B}^*|\}} \kappa_i$'s we can deduce for all modes $i\in\{1,...,2^m\}$ we have
 
  \begin{align}
 \|x_t\|\leq \kappa^* (1-\gamma^*/2)^t\|x_{\tau_q}\|+\frac{2\kappa^*}{\gamma^*} \max\limits_{\tau_q\leq t\leq \tau_{q+1} }\|w_t\| \label{eq:OneeShotmodeepoch2}
 \end{align}
 where $\kappa^*=\sqrt{{2\bar{\nu}}/{\alpha_0\sigma_{\omega}^2}}$ and $\bar{\nu}:=\max_{i\in \{1,..., 2^m\}}$.

Considering
 \begin{align*}
\|w_t\|\leq 10\sigma_{\omega}\sqrt{n\log\frac{T}{\delta}}
\end{align*}
we can upper-bound second term on right hand side of (\ref{eq:OneeShotmodeepoch2}) by $\bar{L}$ which is defined as follows:
 \begin{align*}
 \bar{L}:=\frac{20\kappa^*}{\gamma^*}\sigma_{\omega}\sqrt{n\log\frac{T}{\delta}}. 
\end{align*}
Let us now denote the time sequences in which the switches in actuating mode happens be $I=\{k_0,k_1,...,k_{l-1},k_{l}\}$, then one can upper bound the state norm as follows:
\begin{align*}
\|x_{k_f}\|&\leq (1-\gamma^*/2)^{k_f-k_l}\kappa^*\|x_{k_l}\|+\bar{L}\\
&\leq (1-\gamma^*/2)^{k_l-k_{l-1}} (1-\gamma^*/2)^{k_f-k_l}\kappa^{*2} \|x_{k_{l-1}}\|\\
&+\bar{L} (1-\gamma^*/2)^{k_l-k_{l-1}}\kappa^{*}+\bar{L}\\
&\leq  ... \leq (1-\gamma^*/2)^{k_f-0}\kappa^{*(k_f-0)/\tau}\|x_0\|+\Omega
\end{align*}

in which $\gamma^*=\frac{1}{2\kappa^{*}}$. The term $\Omega$ includes all $\bar{L}-$ dependent terms. 
 
For stability purpose, the following inequality
\begin{align}
\kappa^{*}(1-\gamma^*/2)^{\tau_{AD}}=1-\chi<1 
\end{align}
 needs to be fulfilled which yields
\begin{align}
\tau_{AD}>\tau_{MADT}:= Ceil\bigg[\frac{\ln \kappa^*- \ln(1-\chi)}{-\ln(1-\gamma^*)}\bigg]
\label{eq:tau_a_slow2} 
\end{align}
in which $\tau_{MADT}$ is MADT and  $0<\chi<1$ is a user-define parameter.
Subsequently we can write:
\begin{align}
\Omega\leq \bar{L} \sum_{i=0}^{\infty}\big((1-\gamma^*/2)\kappa^{*1/\tau_{AD}}\big)^i\leq \frac{ \bar{L}}{\chi}=:U_{\Omega}
\end{align}

Then, the state is upper-bounded by:
\begin{align}
\|x(t)\|&\leq (1-\chi)^t\|x_0\|+U_{\Omega}\leq 
e^{-\chi t}\|x_0\|+U_{\Omega}\label{eq:upper_bound_x_slow}
\end{align}
\end{proof}

The boundedness of $V^i_t$ at each epoch is very crucial for estimation purpose. To bound it, we need to ensure that the extended state $z^i_t$ stays bounded as the sequentially stabilizing policy is executed. The following lemma gives this upper-bound and defines the parameter $\bar{\Upsilon}$ which are used in regret bound analysis section. 
\begin{lemma} \label{define_parBeta}
	Applying the Algorithms 1 and 2 yields 
	
	\begin{align*}
	&\sum_{s=1}^{t}||z(s)||^2\leq 2\bar{\Upsilon} T\;\;\; and \;\;
	||V_t||\leq (1+2\bar{\Upsilon})T\\
& \bar{\Upsilon}=\frac{3200{\kappa^*}^6\sigma^2n}{\chi^2}\log \frac{T}{\delta}
	\end{align*}
\end{lemma}
\begin{proof}
Proof can be found in Appendix \ref{partF}.
\end{proof}

\subsection{Regret Bound Analysis}\label{regretBoundAnalysis}
This section presents regret bound analysis of the proposed "learn and control while switching strategy." Proceeding as in  ~\cite{cohen2019learning}, the regret can be decomposed into the form of (\ref{eq:R1}-\ref{eq:R4}).  However, each term must be bounded differently than in ~\cite{cohen2019learning} due to the switching nature of the closed-loop system. 

The following lemma gives an upper bound on the ratio of covariance matrix determinant of the central ellipsoid and that of its projection to an arbitrary space associated with mode $i$. The result will be useful in upper-bounding $R_1$ and $R_4$ terms.
\begin{lemma} \label{frac_of_determinant_log}
Given $V_{[0\; t]}$ and $V^{pi}_{[0\; t]}$ as the covariance matrices associated with central ellipsoid and its projection to the space of mode $i$ respectively, then for any $t\leq T$
\begin{align}
\log\frac{\det(V_{[0\; t]})}{\det(V^{pi}_{[0\; t]})}\leq (m-1) \log ((1+2\bar{\Upsilon})t)\label{eq:lemm11ratio}
\end{align}
holds true almost surely.
\end{lemma}
\begin{proof}
Proof has been provided in Appendix \ref{partF}.
\end{proof}

Following Lemma gives an upper bound for the term $R_1$. 
\begin{lemma} \label{Thm.BoundR1}
	Let $ns$ denote number of epochs.The upper bound for $|R_1|$ is given as follows:
	\begin{align}
	|R_1|\leq \frac{\bar{\nu} (1+N_{ts})}{\sigma_{\omega}^2}X_{[0\;\;T]}
	\end{align}	
	where $X_{[0\;T]}=\max_{t\leq T}\|x_t\|$ to be defined by
\begin{align}
X_{[0\;\;T]}^2= 2(1-\chi)^{2t}||x_0||^2+\frac{800 {\kappa^*}^2 \sigma_{\omega}^2  n}{\chi^2{\gamma^*}^2}\log\frac{T}{\delta}  
\label{eq:BoundXXX}
\end{align}
	
	and 
\begin{align*}
N_{ts}:=(1+2\bar{\Upsilon})\log T+(ns-1)(m-1)\log(1+2\bar{\Upsilon})T.
\end{align*}
	\begin{align*}
	\bar{\Upsilon}=\frac{3200{\kappa^*}^6\sigma_{\omega}^2n}{\chi^2}\log \frac{T}{\delta}
	\end{align*}	
\end{lemma}

\begin{proof}
Proof can be found in Appendix \ref{partF}.
\end{proof}

\begin{lemma} \label{R2}
	The upper bound on $|R_2|$ is given as follows: 
	\begin{align*}
	|R_2|\leq \frac{\bar{\nu}\vartheta}{\sigma_{\omega}}\sqrt{3\bar{\Upsilon} T\log\frac{4}{\delta}}
	\end{align*}	
which holds true with the probability at least $1-\delta/4$
\end{lemma}
\begin{proof}
	Proof directly follows \cite{cohen2019learning} by taking into account $\|P(\hat{\Theta}^{i}_t,Q,R^{i})\|\leq \bar{\nu}/\sigma^2 \;\; \forall i\in \mathcal{I}$ where $\bar{\nu}=\max_{i} \mu_i$. 
\end{proof}

To bound $R_3$ we have the following lemma:
\begin{lemma}\label{thm:boundR3}
	The term $R_3$ has the following upper bound when the sequence of switches is $1,...,j$ and number of epochs is $d$
	\begin{align}
	\nonumber R_3=\sum _{t=\tau_1=0}^{\tau_2} &\big(w^T_tP(\hat{\Theta}^{\sigma(t)}_t,Q,R^{\sigma(t)})\omega_t\\
	\nonumber &-\sigma_{\omega}^2\|P(\hat{\Theta}^{\sigma(t)}_t,Q,R^{\sigma(t)})\|_*\big)1_{\mathcal{E}_t}+...\\
\nonumber	&+\sum _{t=\tau_{ns}}^{T} \big(w^T_tP(\hat{\Theta}^{\sigma(t)}_t,Q,R^{\sigma(t)})\omega_t\\
\nonumber &-\sigma_{\omega}^2\|P(\hat{\Theta}^{\sigma(t)}_t,Q,R^{\sigma(t)})\|_*\big)1_{\mathcal{E}_t}\\
&\leq 8\bar{\nu}\sqrt{T\log^3\frac{4T}{\delta}}
	\end{align}
\end{lemma}
\begin{proof}
	Proof directly follows \cite{cohen2019learning} by taking into account $\|P(\hat{\Theta}^{i}_t,Q,R^{i})\|\leq \bar{\nu}/\sigma^2 \;\; \forall i\in \mathcal{I}$ where $\bar{\nu}=\max_{i} \mu_i$. 
\end{proof}	

\begin{lemma}\label{thm:boundR4}
	The term $R_4$ has the following upper bound.
	
	\begin{align}
		\nonumber R_4 \leq  \frac{8\bar{\nu}}{\sigma_{\omega}^2}\big(&\bar{r}+(1+\bar{r})\vartheta \sqrt{(1+2\bar{\Upsilon})T}\big)\big((1+2\bar{\Upsilon})\log T\\
	&+(ns-1) \log(1+2\bar{\Upsilon})T\big)\\
	\nonumber 	\bar{r}:= 8\sigma_{\omega}^2n \bigg(&2\log\frac{n}{\delta}+(1+2\bar{\Upsilon})\log T+\\
\nonumber &(m-1)\log(1+2\bar{\Upsilon})T\bigg)+2\epsilon^2 \lambda 
	\end{align}

\end{lemma}
\begin{proof}
Proof can be found in Appendix \ref{partF}.
\end{proof}

\textbf{Proof of Theorem \ref{thm:RegretBound}}

\begin{proof}

\textbf{Proof of first statement}

 Putting everything together gives the following upper bound for regret $R_T$  
\begin{align*}
R_T&\leq \frac{\bar{\nu} (1+N_{ts})}{\sigma_{\omega}^2}X_{[0\;\;T]}+\\
&\frac{\bar{\nu}\vartheta}{\sigma_{\omega}}\sqrt{3\bar{\Upsilon} T\log\frac{4}{\delta}}+8\bar{\nu}\sqrt{T\log^3\frac{4T}{\delta}}+\\
&\frac{8\bar{\nu}}{\sigma_{\omega}^2}\big(\bar{r}+(1+\bar{r})\vartheta \sqrt{(1+2\bar{\Upsilon})T}\big)\big((1+2\bar{\Upsilon})\log T+\\
&(ns-1) \log(1+2\bar{\Upsilon})T\big)
\end{align*}
	where 
	\begin{align}
X_{[0\;\;T]}^2= 2(1-\chi)^{2t}||x_0||^2+\frac{800 {\kappa^*}^2 \sigma_{\omega}^2  n}{\chi^2{\gamma^*}^2}\log\frac{T}{\delta}  
	\end{align}
\begin{align*}
N_{ts}:=(1+2\bar{\Upsilon})\log T+(ns-1)(m-1)\log(1+2\bar{\Upsilon})T.
\end{align*}

 \textbf{Proof of second statement}
 
 	The regret bound of an OFU-based algorithm is directly connected to the extent of tightness of confidence ellipsoid. We can explicitly see this effect in the term $R_4$. This being said, we now compare the confidence ellipsoid built by the projection-based algorithm and the one obtained by naively and repeatedly applying the basic OFU-based algorithm.

Assume at an arbitrary time $\tau$ system actuates in any arbitrary mode $i$, for which the confidence ellipsoid is given by 
\vspace{-0.3cm}
\begin{align}
\nonumber \mathcal{C}^i_t(\delta):=\big\{&\Theta^i \in \mathbb{R}^{(n+m_i)\times n}|\\ &\operatorname{Tr}((\Theta^{i}-\hat{\Theta}^i_t)^\top V^i_{[0\;\; t]}(\Theta^{i}-\hat{\Theta}^i_t))\leq r^{i}\}
\end{align}
where
 \begin{align}
\nonumber r^i =\bigg(& \sigma_{\omega} \sqrt{2n \log \frac{n\det(V^i_{[0\;\; t]}) }{\delta \det(V_{[0,\; \tau]}^{pi})}}+\sigma_{\omega} \sqrt{2n\log \frac{n \det(V_{[0,\; \tau]})}{\delta \det (\lambda I)}}+\\
&\sqrt{\lambda} \epsilon\bigg)^2. 
\end{align}
and $V^i_{[0\;\; t]}$ and $\hat{\Theta}^i_t$ is given by (\ref{eq:obswpe},\ref{computation}). 

On the other when the projection is not applied the confidence set is given as follows:
 \begin{align}
\nonumber \mathcal{\bar{C}}^i_t(\delta):=\big\{&\Theta^i \in \mathbb{R}^{(n+m_i)\times n} |\\
&\operatorname{Tr}((\Theta^{i}-\hat{\bar{\Theta}}^i_t)^\top \bar{V}^i_{t}(\Theta^{i}-\hat{\bar{\Theta}}^i_t))\leq r_n^{i}\}
\end{align}
 where 
  \begin{align}
 \hat{\bar{\Theta}}_{t} &={\bar{V^i}_{t}}^{-1}\big({Z_t^i}^\top X_t+ \lambda \Theta^i_0\big). 
 \end{align}
and
 \begin{align*}
V^i_{t}=\lambda I + \sum_{s=1}^{ns_{t-1}} z^i_{s}{z^i_{s}}^T,\;\;\; {Z_t^i}^\top X_t=\sum_{s=1}^{ns_{t}} z^i_{s}{x_s^i}^\top
 \end{align*}
 where $ns^i_t$ is number of time steps that system has been actuating in mode $i$ by time $t$, and   
  \begin{align}
r_n^i &=\bigg( \sigma_{\omega} \sqrt{2n \log \frac{n\det(\bar{V}^i_{t}) }{\delta \det (\lambda I_{(n+m_i)\times (n+m_i)})}}+\sqrt{\lambda} \epsilon\bigg)^2. 
\end{align}
Note that $\Theta_*^i$ belongs to the both confidence sets $\mathcal{C}(\delta)$ and $\mathcal{\bar{C}}(\delta)$ with same probability of $1-\delta$. Hence, the tightest confidence set gives lower regret. Considering the fact that $V^i_{[0,\; t]}\succeq \bar{V}_t$, there exists a $t_0<T$ for big enough $T$ such that for $t>t_0$
\begin{align*}
    \det (\frac{V^i_{[0,\; t]}}{r^i})>\det (\frac{\bar{V}_t^i}{r_n^i}).
\end{align*}
This guarantees $\mathcal{C}(\delta)\subset \mathcal{\bar{C}}(\delta)$ meaning that the confidence ellipsoid $\mathcal{C}(\delta)$ is tighter.
 \end{proof}


\clearpage
 \section{Supplementary Proofs}
\label{partF}
\textbf{Proof of Lemma \ref{Self_normalized_Bound_1}}

\begin{proof} 
The proof follows from Theorem 3 and Corollary 1 in \cite{abbasi2011regret} which gives.
\begin{align}
\left\lVert S^i_t\right\rVert^2_{V_{t}^{-1}} \leq 2 \sigma_{\omega}^2 \log\bigg(\frac{\det(V_t)^{1/2} \det(\lambda I)^{-1/2}}{\delta}\bigg)
\end{align}	
we have 
\begin{align*}
\log\bigg(\frac{\det(V_t)}{ \det(\lambda I)}\bigg)\geq \log\bigg(\frac{\det(V_t)^{1/2}}{ \det(\lambda I)^{1/2}}\bigg),
\end{align*}	
true which completes proof.
\end{proof}

\textbf{Proof of Corollary \ref{Self_normalized_Bound}}
\begin{proof}
Applying Lemma \ref{Self_normalized_Bound} for $i=1,...,n$ yields
\begin{align}
\left\lVert S^i_t\right\rVert^2_{V_{t}^{-1}} \leq 2 \sigma_{\omega}^2 \log \bigg(\frac{n\det(V_t) }{\delta \det(\lambda I)}\bigg)
\end{align}
with probability at least $1-\delta/n$. Furthermore, we have
\begin{align*}
\operatorname{Tr} (S^T_tV^{-1}_tS_t)=\sum_{i=1}^{n}S^{iT}_tV^{-1}_tS^{i}_t
\end{align*}
using which and applying union bound (\ref{eq:mart_vec}) holds with probability at least $1-\delta$.
\end{proof}

\textbf{Proof of Lemma \ref{thm:compact_ellips}}

\begin{proof} 
Given definition of $r^i$ in (\ref{eq:confidence_set_noInfot}), one can write:
 	\begin{align*}
     &r^i =
     \bigg(\sigma_{\omega}\big[\sqrt{2n\log\frac{n\det({V^i_{[0,t]}})}{\delta\det(V^{pi}_{[0,\tau_q]})}}+\sqrt{2n\log\frac{n\det(V_{[0,\tau_q]})}{\delta\det(\lambda I_{(n+m)})}}
     \big]\\
     &+\sqrt{\lambda}\vartheta\bigg)^2\leq \\
     &
    8\sigma_{\omega}^2n \bigg(\log\frac{n\det({V^i_{[0\;\; t]}})}{\delta\det(V^{pi}_{[0\;\; \tau_q]})}+\log\frac{n\det(V_{[0\;\; \tau_q]})}{\delta\det(\lambda I_{(n+m)})}\bigg)+\\
    & 2\lambda \vartheta^2\leq \\
    & 8\sigma_{\omega}^2n \bigg(\log\frac{n\det({V^i_{[0\;\; t]}})}{\delta\det(\lambda I_{(n+m_i)})}+\log\frac{n\det(V_{[0\;\; \tau_q]})}{\delta\det(V^{pi}_{[0\;\; \tau_q]})}\bigg)+\\
    &2 \lambda \vartheta^2\leq \\
   &8\sigma_{\omega}^2n \bigg(2\log\frac{n}{\delta}+\log\frac{\det({V^i_{[0\;\; t]}})}{\det(\lambda I_{(n+m_i)})}+\log\frac{\det(V_{[0\;\; \tau_q]})}{\det(V^{pi}_{[0\;\; \tau_q]})}\bigg)+\\
   &2\lambda\vartheta^2
              \end{align*}	
where in the second inequality we apply $\det(\lambda I_{(n+m_i)})<\det(\lambda I_{(n+m)})$. The proof is complete by applying the following properties
 \begin{align}
\log\frac{\det({V^i_{[0\;\; t]}})}{\det(\lambda I_{(n+m_i)})}\leq (1+2\bar{\Upsilon})\log t \label{eq:boundFullactRatio}\\
 \log\frac{\det(V_{[0\;\; \tau_q]})}{\det(V^{pi}_{[0\;\; \tau_q]})}\leq (m-1)\log ((1+2\bar{\Upsilon}))\tau_q  \label{eq:boundjumpdet}
\end{align}
where (\ref{eq:boundFullactRatio}) can be shown by following the same steps of Lemma 28 of \cite{cohen2019learning} and applying $\sum_{s=1}^{t}||z(s)||^2\leq 2\bar{\Upsilon} T$. And the inequality (\ref{eq:boundjumpdet}) is followed by (\ref{eq:lemm11ratio}).
 \end{proof}

\textbf{Proof of Lemma \ref{upper_ and lower bound of P}}

 \begin{proof} 
	From Lemma 19 of \cite{cohen2019learning} we have:
\begin{align*}
   &P_*(\Theta^i_*, Q, R^i)- P(\hat{\Theta}^i_t, Q, R^i) \preceq \frac{2 \kappa_i^2\mu_i}{\gamma_i}\|P(\Theta^i_t, Q, R^i)\|_*\\
   & \lVert\begin{pmatrix} I \\ K(\hat{\Theta}^i_t, Q, R^i) \end{pmatrix}{V^i}^{-1}_{[0,\;t]}\begin{pmatrix} I \\ K(\hat{\Theta}^i_t, Q, R^i) \end{pmatrix}^\top\rVert I\preceq \\
   & \frac{2 \kappa_i^2\mu_i}{\gamma_i}\|P(\Theta^i_t, Q, R^i)\|_* \big (2 \kappa_i^2 \|{V^i}^{-1}_{[0,\;t]}\|\big) I\preceq\\
   & \frac{2 \kappa_i^2\mu_i}{\gamma_i}\|P(\Theta^i_t, Q, R^i)\|_* (2 \kappa_i^2 \sqrt{n+m}){V^i}^{-1}_{[0,\;t]} \preceq\\
   & \frac{2 \kappa_i^2\mu_i}{\gamma_i}\|P(\Theta^i_t, Q, R^i)\|_* (2 \kappa_i^2 \sqrt{n+m})\frac{1}{\lambda_i}\preceq \frac{\alpha_0^i\gamma_i}{2}I
\end{align*}
where the last inequality holds by $\lambda\geq 8 \kappa_*^8 \sqrt{n+m} \frac{4 \bar{\nu} \bar{\mu}}{\alpha_0^i \sigma_{\omega}^2}$.
	\end{proof}

\textbf{Proof of Lemma \ref{define_parBeta}}

\begin{proof} 
	Having (\ref{eq:upper_bound_x_slow}) and given the fact that $z_t=\begin{pmatrix} I \\ K \end{pmatrix}x_t$ and  $||\begin{pmatrix} I \\ K \end{pmatrix}||^2\leq2{\kappa^*}^2$, we have
	\begin{align}
	||z(t)||^2&\leq 4 
	{\kappa^*}^2(1-\chi)^{2t}||x_0||^2+\frac{1600 \sigma^2 {\kappa^*}^4 n}{\chi^2{\gamma^*}^2}\log\frac{T}{\delta} \label{upperBoundZ} 
	\end{align}
	which yields 
	\begin{align*}
	\sum_{s=1}^{t}\|z(s)\|^2&\leq \frac{4 {\kappa^*}^2}{2\chi-\chi^2}
	\|x_0\|^2+\frac{3200{\kappa^*}^6\sigma_{\omega}^2n}{\chi^2}\log \frac{T}{\delta}t\\
	&\leq \frac{4 {\kappa^*}^2}{\chi}
	t+\bar{\Upsilon} t\leq 2\bar{\Upsilon} t
	\end{align*}
	where 
	\begin{align*}
	\bar{\Upsilon}=\frac{3200{\kappa^*}^6\sigma_{\omega}^2n}{\chi^2}\log \frac{T}{\delta}
	\end{align*}
	and in second inequality we applied $t\geq x_0$.
	
	Having $\sum_{s=1}^T||z_s||^2\leq 2\bar{\Upsilon} T$ along with $T\geq \lambda$ which is true always by recalling the definition of $\lambda$, regardless of any sequences of actuating modes, it returns out that
	\begin{align*}
	||V_t||\leq \lambda+\beta^{-1}\sum_{s=1}^t||z_s||^2\leq (1+2\bar{\Upsilon})T
	\end{align*}
	by picking $\beta=1$. 
\end{proof}

\textbf{Proof of Lemma \ref{frac_of_determinant_log}}

\begin{proof} 
To upper-bound, the ratio (\ref{eq:lemm11ratio}) where the nominator is the covariance of the central ellipsoid and the denominator is the covariance of the ellipsoid constructed by the projection of the central ellipsoid to the space of mode $i$, one can write:
\begin{align}
\frac{\det(V_{[0,\;t]})}{\det(V^{pi}_{[0,\;t]})}=\frac{\prod_{l=1}^{m_i}\frac{1}{\lambda^i_l}}{\prod_{k=1}^{m}\frac{1}{\lambda_k}} \label{eq:bunch_of_lambda}
\end{align}
where $\lambda_k$ for $k \in \{1,...,n+m\}$ and $\lambda^i_l$ for $l\in \{1,...,n+m_i\}$ are eigenvalues of $V_{[0,\;t]}$ and $V^i_{[0,\;t]}$ respectively. We know that for any ellipsoid constructed by the covariance matrix $V_{[0,\;t]}$, we have $1/\lambda_k$'s as the diameters of ellipsoid where $\lambda_k$'s are corresponding eigenvalues of $V_{[0,\;t]}$. Since $V^{pi}_{[0,\;t]}$ is the covariance matrix of an ellipsoid which is projection of the central ellipsoid with covariance matrix $V_{[0,\;t]}$, there always exist a $s \in \{1,...,n+m\}$ such that $1/\lambda_s>1/\lambda^i_l$ for all $l\in \{1,...,n+m_j\}$. In other words, there is always a diameter of the central ellipsoid associated with $V_{[0,\;t]}$ which is greater than or equals to any diameter of its projection defined by $V^{pi}_{[0,\;t]}$.  Using this fact we can upper bound (\ref{eq:bunch_of_lambda}) as follows
\begin{align*}
\frac{\det(V_{[0,\;t]})}{\det(V^{pi}_{[0,\;t]})}\leq \lambda_s^{m-m_i}&\leq \lambda_{max}^{m-m_i}\\
&\leq \lambda_{max}^{m-1}(V)\leq (1+2\bar{\Upsilon})^{m-1}t^{m-1}
\end{align*}
The last inequality directly follows by the fact that $(m-m_i)$ takes its maximum value of $m-1$ when the projection is performed into space with dimension $n+1$ (i.e., projection from central ellipsoid to the mode with only one actuator). Also, we have have $\lambda_{max}(V_{[0,\;t]})\leq \|V_{[0,\;t]}\|\leq (1+2\bar{\Upsilon})T$ which completes proof. 
\end{proof}

\textbf{Proof of Lemma \ref{Thm.BoundR1}}

\begin{proof}
	Recalling the switch time sequence $\{\tau_1,...,\tau_{ns}\}$ with $\tau_1=0$ and $T>\tau_{ns}$, we have $ns$ number of epochs.

	Considering (\ref{eq:R1}), one can write
	\begin{align}
	\nonumber R_1&=\sum_{t=0}^{T}(x^T_tP(\hat{\Theta}^{\sigma(t)}, Q, R^{\sigma(t)})x_t-\\
	\nonumber &x^T_{t+1}P(\hat{\Theta}^{\sigma(t)}, Q, R^{\sigma(t)})x_{t+1})1_{\mathcal{E}_t}=\\
	\nonumber&\sum_{t=1}^{T}(x^T_t(P(\hat{\Theta}^{\sigma(t)}, Q, R^{\sigma(t)})-\\
	\nonumber &P(\hat{\Theta}^{\sigma(t-1)}, Q, R^{\sigma(t-1)}))x_t) 1_{\mathcal{E}_t}+\\
	&x_0^TP(\hat{\Theta}^{\sigma(0)}, Q, R^{\sigma(0)})x_0-x_{T+1}^TP(\hat{\Theta}^{\sigma(T)}, Q, R^{i(T)})x_{T+1}.
	\end{align}
	where $i(t)\in \mathcal{I}$. Note that the first term of the right-hand side takes zero value except when there is a switch in either the policy or actuation mode. Therefore, We first need to find an upper bound on the total number of switches.
	
	We know for the system actuating in any arbitrary mode $i$ and arbitrary epoch $k$ happening in time interval $[\tau_k\;\tau^{k+1}]$ we have 
	\begin{align}
		\det(V_{[0\;\tau_k^f]}^{i})\geq 2^{d_k}\det(V_{[0\;\tau_k^0]}^{pi})  \label{eq:how_switch}
	\end{align}
	where $d_k$ is the total number of switches in the epoch.
	Suppose in the last epoch the system actuates in the mode $i$. Furthermore, without loss of generality let system actuates in the mode $j$ within the epoch $ns-1$. Then by applying (\ref{eq:how_switch}) one can write
	\begin{align}
	\nonumber \det(V_{[0,\;T]}^{i})\frac{\det(V_{[0\;\tau_{ns}]})  }{\det(V_{[0,\;\tau_{ns}]}^{pi})  }&\geq 2^{d_{ns}}\det(V_{[0\;\tau_{ns}]})\\
	&\geq 2^{d_{ns}}\det(V^j_{[0,\;\tau_{ns}]}) 
	\end{align}
where the last inequality is due to the fact that $\lambda>1$ and $V^i\;\; \forall i$ is strictly positive definite that guarantees $\det(V_{[0,\;\tau_{ns}]}) \geq\det(V^{j}_{[0,\;\tau_{ns}]})$. By following the same argument backward in switch times,it yields
\begin{align}
\nonumber&\det(V^i_{[0,\;T]})\frac{\det(V_{[0,\;\tau_{ns}]})  }{\det(V_{[0,\;\tau_{ns}]}^{pi})  }... \frac{\det(V_{[0,\;\tau_2]})  }{\det(V_{[0,\;\tau_2]}^{pl})  }\geq\\
&
 2^{d_{ns}+d_{ns-1}+...+d_2}\det(V_{[0,\;\tau_{2}]})
\end{align}
Applying $\det(V_{[0,\;\tau_{2}]}) \geq 2^{d_1}\det(\lambda I)$ results in

\begin{align}
&\nonumber\det(V^i_{[0,\;T]})\frac{\det(V_{[0,\;\tau_{ns}]})  }{\det(V_{[0,\;\tau_{ns}]}^{pi})  }... \frac{\det(V_{[0,\;\tau_2]})  }{\det(V_{[0,\;\tau_2]}^{pl})  }\geq\\
\nonumber &
 2^{d_{ns}+d_{ns-1}+...+d_1}\det(\lambda I_{(n+m)})\geq\\
 &2^{d_{ns}+d_{ns-1}+...+d_1}\det(\lambda I_{(n+m_i)}) 
\end{align}
By taking $\log_2$ on both sides one can write
\begin{align}
\nonumber d_{ns}+d_{ns-1}+...+d_1\leq& \log \frac{\det(V^i_{[0\;T]})}{\det(\lambda I_{(n+m_i)})}\\
&+(ns-1)\mathcal{F}
\end{align}

where
\begin{align}
\mathcal{F}=\max_{\forall i,j} \mathcal {F}_i \label{eq:definMathcalF}
\end{align}
where $i$ denote the actuating mode and $j$ denotes the epoch and $\mathcal{F}_{i,j}$ is defined as follows
\begin{align}
\mathcal{F}_{i,j}=\frac{\det(V_{[0\;\tau_j]})  }{\det(V_{[0\;\tau_j]}^{pi})}:=(m-1)\log(1+2\Upsilon_i)\tau_j \label{eq:definMathcalF_i}
\end{align}
which is result of the Lemma \ref{frac_of_determinant_log}. Then it yields,
\begin{align}
    \mathcal{F}=(m-1)\log (1+2\bar{\Upsilon})T
\end{align}
Furthermore, we have 
\begin{align}
\log(\frac{\det(V^i_{[0\;T]})}{\det(\lambda I_{(n+m_i)})})\leq  (1+2\bar{\Upsilon})\log T \label{eq:ratioVdetoverLam}
\end{align}
in which $\bar{\Upsilon}=\max_i\Upsilon_i$. As a result, the maximum number of switch is upper-bounded by
\begin{align*}
N_{ts}:=(1+2\bar{\Upsilon})\log T+(ns-1)(m-1)\log(1+2\bar{\Upsilon})T.
\end{align*}
where $N_{mts}$ is the maximum number of total switches including the switches between actuating modes. The upper bound for $R_1$ can be written as follows:

	\begin{align}
|R_1|\leq \frac{\nu (1+N_{ts})}{\sigma_{\omega}^2}X_{[0\;\;T]}
\end{align}
which is obtained noting $||P^i_t||\leq \bar{\nu}/\sigma^2$ where $\bar{\nu}=\max_{\forall i}{\nu_i}$ and the upper-bound of state (\ref{eq:BoundXXX}).
\end{proof}

\textbf{Proof of Lemma \ref{thm:boundR4}}

\begin{proof} 
Suppose the system evolves in a sequence of $d$ epochs starting with actuating mode 1, e.g $1,j,...,i$. Then, one can write
	\begin{align}
	&\nonumber\sum _{t=0}^{T} \big({z^{\sigma(t)}_t}^\top{V^{\sigma(t)}_t}^{-1}{z^{\sigma(t)}_t}\big)1_{\mathcal{E}_t}= \sum _{t=0}^{\tau_2} \big(z^T_t{V}^{-1}_{[0\;t]}z_t\big)1_{\mathcal{E}_t}+\\
	\nonumber&\sum _{t=\tau_2}^{\tau_3} \big({z^j}^T_t{V^j}^{-1}_{[0\;t]}z^j_t\big)1_{\mathcal{E}_t}+...+\sum _{t=\tau_{ns-1}}^{\tau_{ns}} \big({z^l}^T_t{V^l}^{-1}_{[0\;t]}z^l_t\big)1_{\mathcal{E}_t}+\\&\sum_{t=\tau_{ns}}^{T} \big({z^i}^T_t{V^i}^{-1}_{[0\;  t]}z^i_t\big)1_{\mathcal{E}_t}
\leq 2\log\frac{\det(V_{[0,\;\tau_2]})}{\det(\lambda I_{(n+m)})}+ \label{eq:Lemma26[25]}\\
\nonumber &2\log\frac{\det(V^j_{[0,\;\tau3]})}{\det(V^{pj}_{[0,\;\tau_2]})}+...+2\log\frac{\det(V^l_{[0\;\tau_{ns}]})}{\det(V^{pl}_{[0\;\tau_{ns-1}]})}+\\
&2\log\frac{\det({V^i}_{[0,\;T]})}{\det(V^{pi}_{[0,\;\;\tau_{ns}]})} 
	= 2\log\frac{1}{\det(\lambda I_{(n+m)})}+\label{eq:telescop}\\
	\nonumber &
	2\log\frac{\det(V_{[0,\;\tau_2]})}{\det(V^{pj}_{[0,\;\tau_2]})}+...+2\log\frac{\det(V^l_{[0\;\tau_{ns}]})}{\det(V^{pi}_{[0,\;\;\tau_{ns}]})}+\\
	&2\log\det({V^i}_{[0,\;T]}) \leq \label{eq:applyproj}\\
\nonumber & 2\log\frac{\det({V^i}_{[0,\;T]})}{\det(\lambda I_{(n+m_i)})}+2\log\frac{\det(V_{[0,\;\tau_2]})}{\det(V^{pj}_{[0,\;\tau_2]})}+...\\
\nonumber&+2\log\frac{\det(V_{[0,\;\tau_{ns}]})}{\det(V^{pi}_{[0,\;\;\tau_{ns}]})}\leq\\
&2(1+2\bar{\Upsilon})\log T+2(ns-1)(m-1) \log(1+2\bar{\Upsilon})T\label{eq:finalizedone}
	\end{align}
	where in the inequality (\ref{eq:Lemma26[25]}) we applied the result of Lemma 26 in \cite{cohen2019learning} and then in (\ref{eq:telescop}) we telescoped the summation of logarithmic terms. The inequality (\ref{eq:applyproj}) is due to the facts that $\det(\lambda I_{(n+m)})\geq \det(\lambda I_{(n+m_i)})$ for $\lambda\geq 1$ and $V^i$ is strictly positive definite for all $i\in \mathcal{I}$ that guarantees $\det(V_{[0,\;\tau_{k}]}) \geq \det(V^{pj}_{[0,\;\tau_{k}]})$. Finally the last inequality is direct result of applying (\ref{eq:definMathcalF}-\ref{eq:ratioVdetoverLam}) and (\ref{eq:boundFullactRatio}).

	Recalling definition $r^{i}+(1+r^{i})\vartheta \sqrt{(1+2\Upsilon^i)T}=:\mu_i$ and $\bar{\mu}=\max_{i\in\mathcal{B}^*}\mu_i$ we have the following upper bound for the terms $\frac{4\nu_i\mu_i}{\sigma_{\omega}^2}$:
	\vspace{-0.3cm}
	 \begin{align*}
      \frac{4\nu_i\mu_i}{\sigma_{\omega}^2}\leq \frac{4\bar{\nu}\bar{\mu}}{\sigma_{\omega}^2}\leq \frac{4\bar{\nu}}{\sigma_{\omega}^2}\big(r^{i}+(1+r^{i})\vartheta \sqrt{(1+2\bar{\Upsilon})T}\big)
	 \end{align*}  
	where $r^{i}$'s $\forall i$ are upper-bounded by (\ref{eq:Upperbound_forr^m})
	\vspace{-0.4cm}
	 \begin{align*}
r^{i}\leq \bar{r}:= 8\sigma_{\omega}^2n \bigg(&2\log\frac{n}{\delta}+(1+2\bar{\Upsilon})\log T+\\
&(m-1)\log(1+2\bar{\Upsilon})T\bigg)+2\epsilon^2 \lambda 
 	\end{align*}	
	which completes the proof.
\end{proof}


\end{document}